\documentclass[11pt]{article}

\usepackage[slantedGreek]{mathpazo}
\usepackage[pdftex]{graphicx}
\usepackage{amsfonts}
\usepackage{setspace}
\usepackage{amssymb}
\usepackage{amsthm}
\usepackage{amsmath}
\usepackage{mathtools}
\usepackage{subcaption}
\usepackage{bm}
\usepackage{xcolor}
\definecolor{darkblue}{rgb}{0.0,0.0,0.55}

\usepackage{enumitem}
\usepackage{booktabs}
\usepackage{array}
\usepackage{multirow}
\usepackage{algorithm}
\usepackage{float}
\usepackage{algpseudocode}
\usepackage{placeins}
\usepackage{needspace}
\usepackage[sort,authoryear]{natbib}
\usepackage[colorlinks,citecolor=darkblue,urlcolor=darkblue,linkcolor=darkblue]{hyperref}

\theoremstyle{plain}
\newtheorem{theorem}{Theorem}[section]

\newtheorem{proposition}[theorem]{Proposition}

\theoremstyle{definition}

\theoremstyle{remark}

\newcommand\dring[1]{%
  {
   \mathop{\kern0pt #1}\limits^{%
     \vbox to-1.85ex{
       \kern-2ex
       \hbox to 0pt{\hss\normalfont\kern.1em \r{}\kern-.45em \r{}\hss}%
       \vss
     }%
   }%
  }%
}

\newcommand{\bfh}{\mathbf{h}}

\newcommand{\bfu}{\mathbf{u}}

\newcommand{\bfx}{\mathbf{x}}
\newcommand{\bfy}{\mathbf{y}}

\newcommand{\calG}{{\cal G}}

\newcommand{\calN}{{\cal N}}

\newcommand{\calO}{{\cal O}}
\newcommand{\calP}{{\cal P}}

\newcommand{\calX}{{\cal X}}

\newcommand{{\Pa}}{\text{Pa}}
\newcommand{{\child}}{\text{child}}

\newcommand{\suppsec}[1]{Section~#1 of the Supplementary Material}
\newcommand{\suppsecp}[1]{Section~#1 of the Supplementary Material}

\allowdisplaybreaks
\title{\vspace{-1cm}\baselineskip=20pt\bf Bayesian Treed Vecchia Gaussian Process Emulation via Sequential Monte Carlo}
\date{}
\begin{document}
\maketitle
\baselineskip=15pt
\vspace{-1.75cm}
\begin{center}
Cong Ye and Pulong Ma \\
Department of Statistics, Iowa State University\\ 
2438 Osborn Dr, Ames, Iowa 50011 \\
\hskip 5mm\\
\end{center}

\begin{abstract}
Computer models (or simulators) are used in many scientific domains to understand and predict real-world processes, but they are expensive to run. One way to address this hurdle is to use Gaussian processes (GPs). For complex physical systems such as storm surge, simulator outputs often exhibit heterogeneous structures, and many nonstationary GP methods become impractical with several inputs and tens of thousands of runs. To address these issues, we combine the divide-and-conquer idea of treed GP models with Vecchia approximations in a Bayesian treed Vecchia GP (TVGP) model. The model uses the Bayesian classification and regression tree prior for random recursive partitioning of the input domain. Within each leaf, a Vecchia approximation allows efficient likelihood evaluation. For Bayesian inference, we develop a sequential Monte Carlo algorithm that explores tree space as an alternative to reversible-jump Markov chain Monte Carlo. The method is demonstrated on synthetic examples with two and ten inputs and on a storm-surge simulator with seven inputs and up to 50,649 training runs. TVGP has the smallest prediction error in both synthetic examples and with the larger storm-surge training set, although its predictive intervals undercover in the example with ten inputs and the storm-surge application.

\medskip
\noindent\textbf{Keywords}: Computer experiments; Gaussian process; sequential Monte Carlo; tree methods; Vecchia approximation
\end{abstract}

\section{Introduction}\label{sec:intro}

Storm surge is one of the most dangerous hurricane hazards affecting coastal communities and has caused severe damage and loss of life.
For U.S. Atlantic tropical cyclones during 1963--2012, it is reported that storm surge accounted for 49\% of direct fatalities, while rainfall-induced freshwater flooding accounted for 27\% \citep{Rappaport2014}. \citet{Weinkle2018} estimated that hurricanes striking the continental United States during 1900--2017 caused about \$2 trillion in damage normalized to 2018 values, or under \$17 billion annually.
The National Hurricane Center (NHC) reports that storm surge was the deadliest hazard during Hurricane Ian in the year 2022, claiming 41 lives, including 36 in Lee County, Florida \citep{NHCIan2023}.

Hazard assessments quantify the risk of storm surge by evaluating a large number of synthetic or hypothetical storm scenarios to estimate water levels over broad coastal regions, identify areas exposed to severe inundation, and quantify uncertainty in flood hazard.
Such analyses are directly relevant to coastal flood studies and flood-elevation estimation for applications such as the Federal Emergency Management Agency (FEMA) National Flood Insurance Program.
For example, FEMA flood studies require estimation of surge elevations associated with multiple annual exceedance probabilities.
Post-Katrina federal studies also adopt synthetic-storm methods based on the joint-probability method \citep{Toro2010JPM}.
Related FEMA coastal hazard studies have also relied on numerical modeling of storms and coastal hydrodynamics together with storm climatology analysis and flood-hazard evaluation \citep{Niedoroda2010}.

In these studies, the ADvanced CIRCulation (ADCIRC) model \citep{adcirc,WesterinkEtAl2008}  has been widely used because it can represent surge and wave dynamics over complex coastal domains with substantial physical realism. ADCIRC has been certified by FEMA and has become one of the most popular surge prediction models since Hurricane Katrina in supporting coastal flood hazard studies \citep{FEMA2015, FEMA2017, Niedoroda2010, Cialone2017}. However, even a single simulation run with ADCIRC can be computationally expensive. 
For example, a single high-resolution ADCIRC simulation in the Southwest Florida (SWFL) study takes roughly 2,000 core-hours on a supercomputer \citep{Ma2022PPCokrig}, making uncertainty quantification (UQ) prohibitively costly.

This scientific hurdle has often been  addressed via surrogate modeling since the seminal work by \citet{Sacks1989}.
A well-established surrogate modeling framework is to use a Gaussian process (GP) (\citealp[pp.~15--86]{santner2003design}; \citealp[pp.~143--221]{gramacy2020surrogates}) for emulating computer models, because a GP can provide fast probabilistic approximation to the computer model due to its closed-form predictive distributions and theoretical properties.
For the SWFL study region, the ADCIRC simulator takes the central pressure deficit, scale pressure radius, forward speed, heading, and Holland's $B$ parameter of a storm, together with its landfall location, as input parameters, and generates storm-surge outputs over the course of storm events. The GP modeling framework has been used to emulate storm surges under various application scenarios \citep{Ma2022PPCokrig,Mccrimmon2026, KYPRIOTI2023104231,jia2016} with these storm characteristics as inputs.
However, physical knowledge indicates that the relationship between the storm-surge output and the inputs is highly nonstationary and nonlinear.
Accurately capturing this relationship plays a crucial role in estimating hazard probability over a broad ensemble of plausible storms.
These existing works follow the standard practice of GP modeling, in which a parametric stationary covariance function such as the Mat\'ern class \citep[pp.~48--51]{Stein1999} is used over the whole input domain.
Thus, they do not address the statistical problem of nonstationarity in the input space and the computational issue due to the large number of simulation runs.

This paper aims to emulate the storm-surge computer model by addressing these two issues of spatial nonstationarity and large datasets.  Standard practice of the GP modeling framework has led to the well-known big ``$n$'' problem \citep{Cressie2008,Banerjee2008} because of the high computational cost in performing the Cholesky decomposition for an $n$-by-$n$ covariance matrix. There is a large literature on scalable GP methods for addressing this problem \citep[e.g.,][]{Katzfuss2017MRA, Gramacy2015, Stein2004, Sang2012}; see \citet{Banerjee2017} for a review from a spatial modeling perspective. Vecchia approximation \citep{Vecchia1988,Datta2016,Katzfuss2021} has become one of the most promising methods. Another challenge is that for computer models, the input space often has dimension greater than two, which may pose computational challenges.  Due to the nature of the storm-surge application, it is also essential to capture spatial nonstationarity over the input space in order to obtain good predictive performance.  To capture complex spatial dependency, several nonstationary GP methodologies have been proposed in the literature.
Existing work includes, but is not limited to, (multi-scale) basis-function methods \citep{Cressie2008,Ma2020FGP, Fox2012}, spatial nonstationary covariance
functions \citep{Paciorek2006}, spatial deformation \citep{Schmidt2003}, and partition-based methods \citep{Kim2005,Gramacy2012,Konomi2014,Luo2024, Luo2022}. While modeling spatial nonstationarity is a well-known issue in spatial statistics, many existing methods developed for geostatistical data cannot be straightforwardly adapted to computer experiments with input dimension larger than two, thus leading to a challenging modeling task for the storm-surge application where storm surge is generated over a seven-dimensional input domain.

This paper addresses the nonstationarity issue through a tree-based method, which is closely related to partition-based methods.
These methods assume local GPs in subregions defined by a partition of the input domain, and then combine the local models to obtain a globally nonstationary model.
\citet{Kim2005} use Voronoi tessellation to partition the spatial domain and assume independent stationary GPs on each Voronoi cell.
The treed Gaussian process \citep{Gramacy2012} uses the Bayesian classification and regression tree (CART) prior \citep{Chipman1998,Chipman2002} to partition the domain recursively through binary splits and fits independent GPs at each leaf node.
To improve scalability, \citet{Konomi2014} combine treed GP with full-scale approximation \citep{Sang2012}.
For massive spatial datasets, \citet{Guhaniyogi2018} propose a divide-and-conquer Bayesian strategy that partitions the data into subsets, fits a Bayesian spatial process model on each subset, and then combines the subset posterior distributions to obtain approximate inference and prediction for the full dataset.
Instead of using hard partitions, \citet{Luo2024} use a minimum spanning tree (MST) prior \citep{Teixeira2019} to partition the spatial domain and use soft partitions for prediction to avoid discontinuities near partition boundaries.

While these tree-based methods successfully model spatial nonstationarity in low-dimensional domains, their extension to higher-dimensional input spaces is not straightforward for modeling computer experiments.
Specifically, existing approaches face three challenges when the input space has more than two dimensions and the dataset is large.
First, methods via the MST prior \citep{Luo2024,Teixeira2019} for the tree generating process rely on Prim's algorithm, which constructs trees iteratively.
This requires calculating pairwise distances for graph construction, leading to $\calO(n^2)$ computational cost.
While the MST approach could be well-suited for capturing highly nonstationary behavior in clustered observations in a two-dimensional spatial domain, it becomes computationally prohibitive and algorithmically difficult when the sample size is large and the input dimension exceeds two.
Second, methods utilizing the Bayesian CART prior \citep{Gramacy2012,Konomi2014} handle higher dimensions better by generating recursive, axis-aligned binary partitions; however, they could suffer from heavy computations within leaf nodes in Markov chain Monte Carlo (MCMC) tree search algorithms.
In shallow trees, a large number of observations may fall into a single leaf, rendering the local GP fitting overwhelmingly expensive in the tree growing or pruning operations.
Finally, regardless of the partitioning scheme, standard posterior learning for these tree-based GP models relies on reversible-jump Markov chain Monte Carlo (RJ-MCMC) algorithms to search over tree space, often leading to slow mixing in posterior sampling without well-tuned proposals.

This paper  develops a Bayesian \textit{treed Vecchia GP} (TVGP) model to simultaneously address the computational, modeling, and inferential challenges.
Borrowing the idea from treed GP, we use a divide-and-conquer strategy to recursively partition the input domain through the Bayesian CART prior \citep{Chipman1998,Chipman2002}.
To resolve the local big ``$n$'' problem within each leaf node, we use Vecchia approximation \citep{Katzfuss2021}, which gives linear-cost leafwise likelihood evaluation for fixed conditioning sets of bounded size.
Furthermore, to avoid relying on a single local Markov-chain trajectory through tree space, we develop a sequential Monte Carlo (SMC) algorithm for particle-based posterior inference of tree structures. This strategy is motivated by recent successes of SMC in other tree-based Bayesian nonparametric models \citep{Lakshminarayanan2013,Awaya2024,Ma2025PPT}.
Although we focus on the Bayesian CART prior, the proposed model and inference algorithm can also be adapted to other recursive partitioning schemes such as oblique-axis partitions \citep{Murthy1994} with some effort.

The contributions of this paper are as follows. First, we propose a Bayesian TVGP model that combines random recursive partitioning with Vecchia approximations to capture locally smooth but globally nonstationary structures.
Second, we develop an SMC-based posterior inference algorithm for tree-based GP models that replaces local Markov-chain traversal with breadth-first exploration. Breadth-first traversal processes nodes within each tree level by level, starting at the root. A key challenge in growing the tree involves marginal likelihood calculation for local GP models at each node. We use a Laplace approximation to the leaf marginal likelihood to evaluate the Bayes factors used in tree expansion. We derive analytic expressions for the gradient and Hessian of the integrated likelihood under the Vecchia model and use these derivatives for optimization and the Laplace approximation. The tree grows sequentially, and each step draws the decision at one leaf from the locally optimal proposal, which looks one step ahead \citep{Lin2013}.
Third, through simulation studies and a storm-surge application, we show that TVGP with the proposed SMC algorithm provides a practical alternative to classical tree-based GP methodology in UQ applications where spatial nonstationarity and large numbers of runs must be addressed.

The rest of the paper is organized as follows. Section~\ref{sec: model} presents the TVGP model, and Section~\ref{sc:bayesian_inference} derives the leaf marginal likelihood and introduces the sequential Monte Carlo method for fitting and prediction. Section~\ref{sec: simulation study} uses two synthetic examples to compare SMC with RJ-MCMC, to separate the contributions of partitioning and Vecchia likelihoods, and to compare TVGP with other GP methods. Section~\ref{sec: storm surge} presents the storm-surge application, and Section~\ref{sec: conclusion} concludes the paper. The Supplementary Material contains all proofs, implementation details, additional numerical studies, and the SMC algorithm with its computational complexity analysis.

\section{Bayesian Treed Vecchia Gaussian Process Model} \label{sec: model}

Let $\mathcal{X}\subset \mathbb{R}^d$ denote the input space and $y(\bfx)\in\mathbb{R}$ denote the output of the function $y(\cdot)$ at any input $\bfx \in \mathcal{X}$. Given $n$ data points $\bfy := \big(y(\bfx_1),\ldots,y(\bfx_n)\big)^\top$ over inputs $\{\bfx_i\}_{i=1}^n$, we model the function $y(\cdot)$ through a tree-based Gaussian process that is specified hierarchically. Let $\pi(T)$ denote the Bayesian CART prior \citep{Chipman1998, Chipman2002} for a stochastic binary tree $T$ constructed through recursive partitioning of $\calX$. The terminal nodes of $T$ define a collection of leaf regions, denoted by $\mathcal{L}_T$. Each element $A\in\mathcal{L}_T$ is a subregion of $\calX$, and the leaf regions form a partition of the input space: $\cup_{A\in\mathcal{L}_T}A=\calX$ and $A\cap A'=\emptyset$ for any two distinct leaves $A,A'\in\mathcal{L}_T$. Given $T\sim \pi(T)$, the model for $y(\cdot)$ can be written as
\begin{equation}
y(\bfx)=\sum_{A\in\mathcal{L}_T} f_A(\bfx)\,\mathbb{I}(\bfx\in A), \qquad \bfx\in\mathcal{X},
\label{eq:tgp_piecewise}
\end{equation}
where $f_A(\cdot)$ is the local process associated with leaf region $A$. We assume that $f_A(\cdot)$ is a Gaussian process corrupted with measurement error within $A$. This model combines the idea of the treed Gaussian process (TGP) model \citep{Gramacy2012} and Vecchia approximations \citep{Katzfuss2021, Vecchia1988, Datta2016}. The tree-based structure induces spatial heterogeneity as in TGP, and the Vecchia approximation of the leaf responses allows fast local likelihood evaluation at leaf nodes even in a multidimensional input space.

\subsection{Vecchia GP Prior for $f_A$}\label{sc:vecchiaGPIntro}

Conditional on $T$, the leaf-specific processes and parameters are independent across leaves. For a fixed leaf node $A \in \mathcal{L}_T$, let $\{\bfx_{A,i}\}_{i=1}^{n_A}$ denote the inputs that fall in region $A$, and write $\bfy(A)=\big(y(\bfx_{A,1}),\ldots,y(\bfx_{A,n_A})\big)^\top$ for the corresponding response vector.
Since $\bfx_{A,i} \in A$, we have $y(\bfx_{A,i}) = f_A(\bfx_{A,i})$ under Equation \eqref{eq:tgp_piecewise}.
We model the leaf-specific response surface through $f_A(\bfx) = \bfh(\bfx)^\top\gamma_A + w_A(\bfx) + \varepsilon_A(\bfx)$, $\bfx \in A$,
where $\bfh(\bfx)\in\mathbb R^p$ contains known mean functions with leaf-specific coefficients $\gamma_A\in\mathbb R^p$, $w_A(\cdot)$ is a zero-mean Gaussian process, and $\varepsilon_A(\bfx)$ is an independent nugget term. Specifically, we assume
\begin{align*}
    w_A(\cdot) \sim \calG\calP\!\left(0, \sigma_A^2 r_A(\cdot,\cdot;\phi_A)\right), \\*
    \varepsilon_A(\bfx) \stackrel{\text{ind}}{\sim} N(0, \sigma_A^2 g_A),
\end{align*}
where $\sigma_A^2 > 0$ is a variance parameter, $\phi_A=(\phi_{A,1},\ldots,\phi_{A,d})^\top$ is the $d$-dimensional vector of positive range parameters, and $g_A>0$ is the nugget-to-process variance ratio. With $\theta_A=(\phi_A,g_A)$, the leaf response vector therefore satisfies
\begin{equation*}
\bfy(A)\mid\gamma_A,\sigma_A^2,\theta_A\sim\calN\!\Big(H\gamma_A,\;\sigma_A^2R_A(\theta_A)\Big),
\end{equation*}
where $H=\{\bfh(\bfx_{A,1}),\ldots,\bfh(\bfx_{A,n_A})\}^\top$ is the $n_A\times p$ design matrix of the leafwise mean and $R_A(\theta_A)=\big(r_A(\bfx_{A,i},\bfx_{A,j};\phi_A)\big)_{i,j=1}^{n_A}+g_AI_{n_A}$. All examples in this paper use the constant mean $\bfh(\bfx)=1$, so that $H=\mathbf 1_{n_A}$.
In this work, we use a product Mat\'ern correlation function $r_A(\bfx,\bfx';\phi_A) = \prod_{j=1}^d r_0\!\left(|x_j-x_j'|;\phi_{A,j}\right)$,
with
\[
r_0(h;\phi)
=
\frac{2^{1-\nu}}{\Gamma(\nu)}
\left(\frac{h}{\phi}\right)^\nu
K_\nu\!\left(\frac{h}{\phi}\right),
\]
where $\nu$ is treated as known, $\phi_{A,j}$ is the range parameter in the $j$th input dimension, and $K_\nu(\cdot)$ is the modified Bessel function of the second kind.
When there is no spatial correlation,  this model reduces to the well-known Bayesian CART model \citep{Chipman1998, Chipman2002}.

The full Gaussian likelihood in a leaf is expensive because it requires dense matrix operations on the $n_A \times n_A$ covariance matrix.
The Vecchia approximation \citep{Vecchia1988, Katzfuss2021} provides a simple way to reduce this cost.
For any ordering of the observations in leaf $A$, the joint density can be written as a product of one-dimensional conditional densities,
\[
p(\bfy(A)\mid \gamma_A,\sigma_A^2,\theta_A)
= \prod_{i=1}^{n_A}
p(y_{A,(i)} \mid y_{A,(1)},\ldots,y_{A,(i-1)},\gamma_A,\sigma_A^2,\theta_A),
\]
with the first term interpreted as an unconditional density.
Vecchia's idea is to replace the full conditioning set $\{1,\ldots,i-1\}$ by a much smaller set of previously ordered observations.
If $N_A(i)\subseteq\{1,\ldots,i-1\}$ denotes a conditioning set with at most $m$ points, then
\begin{equation}
\label{eq:vecchia_leaf_factorization}
p(\bfy(A)\mid \gamma_A,\sigma_A^2,\theta_A)
\approx
\prod_{i=1}^{n_A}
p(y_{A,(i)} \mid \bfy_{A,N_A(i)},\gamma_A,\sigma_A^2,\theta_A).
\end{equation}
The approximation is exact when all previous observations are included in $N_A(i)$, and it becomes computationally attractive when $m \ll n_A$.
In our implementation, we compute an approximate maximin-distance ordering of the training inputs once and restrict it to the observations in each leaf by following the idea in \citet{Guinness2021}, which is also the default implementation in \citet{Katzfuss2021}. We choose $N_A(i)$ as the $\min(m,i-1)$ nearest predecessors in this restricted ordering.

For a Gaussian process, each factor in Equation~\eqref{eq:vecchia_leaf_factorization} is univariate normal and can be computed from only a small covariance submatrix.
Let $R_A(\theta_A)$ be arranged in the chosen ordering. The regression coefficients of each observation on its conditioning set, and the corresponding conditional variances, define a unit lower-triangular matrix $B_A$ and a diagonal matrix $D_A=\mathrm{diag}(D_{A,1},\ldots,D_{A,n_A})$ through
\begin{gather*}
(B_A)_{i,i}=1,\qquad
(B_A)_{i,N_A(i)}=-R_{A;i,N_A(i)}\,R^{-1}_{A;N_A(i),N_A(i)},\\*
D_{A,i}=R_{A;i,i}+(B_A)_{i,N_A(i)}\,R_{A;N_A(i),i},
\end{gather*}
with all other entries of $B_A$ equal to zero.
This is the standard Vecchia approximation of \citet{Katzfuss2021} applied to the leaf responses, as in \citet{Sigrist2022}.
It induces the sparse precision matrix
\begin{equation}\label{eq:leaf_vecchia_precision}
\widetilde Q_A = B_A^\top D_A^{-1} B_A,
\end{equation}
and the corresponding approximate covariance matrix
$
\widetilde R_A = B_A^{-1} D_A B_A^{-\top}.
$
Because each row of $B_A$ has at most $m$ off-diagonal nonzero entries, a leafwise likelihood evaluation, and each derivative evaluation of Section~\ref{sc:marginal_likelihood}, uses calculations on conditioning sets of size at most $m$ and has cost linear in $n_A$ for fixed $m$ and parameter dimension.
This sparse representation in Equation~\eqref{eq:leaf_vecchia_precision} is used for every leafwise likelihood, marginal likelihood, and posterior computation in Section~\ref{sc:bayesian_inference}.
The prediction formulas in Section~\ref{sc:prediction} are based on the same construction extended to joint vectors of observed and prediction locations within a leaf.

\subsection{Bayesian CART Prior for $T$}\label{sc:CARTpriorIntro}
Following the Bayesian CART framework \citep{Chipman1998}, we specify a prior distribution over binary trees conditional on the training inputs $\mathbf X$. We write this distribution as $\pi(T\mid\mathbf X)$ and suppress $\mathbf X$ elsewhere when the conditioning is understood.
Its generative process recursively partitions the input domain $\mathcal{X}$ into a set of non-overlapping hyper-rectangles corresponding to the leaf nodes $A$ of the tree $T$.
The probability of any specific tree $T$ is defined as the product of the probabilities of the splitting decisions at the internal nodes and the stopping decisions at the leaf nodes.
For any node $A$ at depth $d_A$ in the tree (where the root has depth zero), the process proceeds as follows:

\paragraph{Splitting probability.}
Let $\mathcal C_j(A)$ be a finite set of candidate cut points on axis $j$ that satisfy the prescribed tree constraints, namely a minimum number of training observations in each child and, when imposed, a maximum depth. In our implementation, the candidates are empirical quantiles of the inputs in $A$, so these sets are determined by the training inputs. Write $\mathcal J(A)=\{j:\mathcal C_j(A)\ne\varnothing\}$ for the axes on which a split is possible. We use the depth-dependent probability $\rho(A)=\alpha(1+d_A)^{-\beta}$, with $\alpha\in(0,1)$ and $\beta\ge0$.
Here, $\alpha$ controls the baseline probability of a split, and $\beta$ determines its decay with depth. Set $\rho^\star(A)=\rho(A)\mathbb I\{\mathcal J(A)\ne\varnothing\}$ and draw $S(A)\sim\operatorname{Bernoulli}(\rho^\star(A))$. Thus, a node with no valid split is terminal with probability one.

\paragraph{Splitting rule.}
Conditional on a split ($S(A)=1$), a splitting rule is generated by selecting a cutting dimension $J(A)$ and a cut point $L(A)$.
We reserve $L(A)$ for the scalar cut location at node $A$ and $\mathcal{L}_T$ for the leaf set of tree $T$.
\begin{enumerate}
    \item \textbf{Dimension selection.} The axis $J(A)$ is drawn from a categorical distribution on $\mathcal J(A)$, with probabilities $\lambda_j(A)$ summing to one over this set. We use $\lambda_j(A)=1/|\mathcal J(A)|$ for $j\in\mathcal J(A)$, where $|\cdot|$ denotes the cardinality of a set.
    \item \textbf{Cut-point selection.} Given the chosen axis $j = J(A)$, the cut point $L(A)$ has the discrete uniform distribution $F_{L,A,j}(\{c\})=P\{L(A)=c\mid S(A)=1,J(A)=j\}=1/|\mathcal{C}_j(A)|$ for $c\in\mathcal{C}_j(A)$. The cut points lie strictly inside the range of node $A$ on axis $j$ and satisfy the tree constraints above.
\end{enumerate}

The joint prior probability of a tree $T$ is the product of the probabilities associated with the variables $\{S(A), J(A), L(A)\}$ for all internal nodes and the stopping probabilities for all leaf nodes:
\begin{equation*}
\pi(T\mid\mathbf X) = \left[ \prod_{A \in \mathcal{I}_T} \rho^\star(A)\lambda_{J(A)}(A)F_{L,A,J(A)}(\{L(A)\}) \right]\prod_{A^\prime \in \mathcal{L}_T} \{1 - \rho^\star(A^\prime)\},
\end{equation*}
where $\mathcal{I}_T$ denotes the set of internal nodes and $\mathcal{L}_T$ denotes the set of leaf nodes. This generative formulation ensures a valid probability distribution over the discrete space of tree structures. Because $\rho(A)$ decreases with depth when $\beta>0$, the prior favors shallow trees.

\section{Bayesian Inference}\label{sc:bayesian_inference}
Our inferential target is the marginal posterior distribution over tree structures, $p(T\mid\bfy,\mathbf X)\propto\pi(T\mid\mathbf X)\prod_{A\in\mathcal{L}_T}\Phi(A)$, where $\Phi(A)$ is the leafwise marginal likelihood of Section~\ref{sc:marginal_likelihood}. Conditional on a tree $T$, the leaf-specific mean coefficients and variance are integrated out analytically under a proper normal--inverse-gamma prior, and the covariance parameters are integrated out by a Laplace approximation that uses the Vecchia representation.

\subsection{Marginal Likelihood}\label{sc:marginal_likelihood}

For a fixed leaf $A$ and conditional on $\theta_A$, the leaf model of Section~\ref{sc:vecchiaGPIntro} and its proper normal--inverse-gamma prior are
\begin{equation}\label{eq:nig_leaf_model}
\begin{gathered}
\bfy(A)\mid\gamma_A,\sigma_A^2,\theta_A\sim N_{n_A}(H\gamma_A,\sigma_A^2R_A),\\
\gamma_A\mid\sigma_A^2\sim N_p(m_0,\sigma_A^2V_0),\qquad
\sigma_A^2\sim\mathrm{IG}(a_0,b_0),
\end{gathered}
\end{equation}
where $R_A=R_A(\theta_A)$ is positive definite, the inverse-gamma distribution uses the shape--scale convention, and the hyperparameters $m_0$, $V_0\succ0$, $a_0>0$, and $b_0>0$ are common to all candidate leaves. Write $Q_A=R_A^{-1}$ and $P_0=V_0^{-1}$, and define
\begin{gather*}
M_A=P_0+H^\top Q_AH,\qquad
h_A=P_0m_0+H^\top Q_A\bfy(A),\\
S_{0,A}^2=\bfy(A)^\top Q_A\bfy(A)+m_0^\top P_0m_0-h_A^\top M_A^{-1}h_A,
\end{gather*}
together with $a_A=a_0+n_A/2$ and $b_A=b_0+S_{0,A}^2/2$.

\begin{proposition}[Integrated leaf likelihood]\label{prop:laplaceapprox}
Under model~\eqref{eq:nig_leaf_model}, integrating out $\gamma_A$ and $\sigma_A^2$ gives
\[
q^I(A;\theta_A)
=\frac{b_0^{a_0}\,\Gamma(a_A)}{(2\pi)^{n_A/2}\,\Gamma(a_0)}\,
|R_A|^{-1/2}\,|V_0|^{-1/2}\,|M_A|^{-1/2}\,b_A^{-a_A}.
\]
\end{proposition}

\vspace{-\parskip}
The proof of Proposition~\ref{prop:laplaceapprox} is given in \suppsec{A.1}. The covariance parameters of a leaf can be weakly identified, particularly when the leaf holds few observations. For isotropic Mat\'ern models with fixed smoothness in at most three dimensions, fixed-domain asymptotic theory shows that the variance and the range cannot both be estimated consistently \citep{Zhang2004}. Integrating out $\sigma_A^2$ in Proposition~\ref{prop:laplaceapprox} removes the variance, but $q^I(A;\theta_A)$ can remain nearly flat as the ranges grow, when the correlation matrix approaches a matrix of ones, or shrink, when it approaches the identity matrix.

We therefore use the jointly robust prior of \citet{Gu2019}, which regularizes the posterior near these two limits. For $\zeta_{A,j}=1/\phi_{A,j}$, its density on $\zeta_{A,j}>0$ and $g_A>0$ is proportional to $t_A^a\exp(-bt_A)$, where $t_A=\sum_{j=1}^dC_{A,j}\zeta_{A,j}+g_A$, $C_{A,j}$ is $n_A^{-1/d}$ times the span of the $j$th input in $A$, $a=0.2$, and $b=1$. We assume that each input has positive span in every candidate leaf, so $C_{A,j}>0$. The normalizing constant of the prior, and the Jacobian for the log coordinates below, are given in \suppsec{B.1}. Let $\pi(\theta_A)$ denote the corresponding density on $\theta_A$, so that the leafwise marginal likelihood is $\Phi(A)=\int_{\Theta_A}q^I(A;\theta_A)\pi(\theta_A)\,d\theta_A$ with $\Theta_A=(0,\infty)^{d+1}$.

We assume that $\Phi(A)$ is finite for every candidate leaf. Because the normal--inverse-gamma and jointly robust priors are proper, $\Phi(A)$ is a probability density in $\bfy(A)$, so the assumption holds for almost every response vector and, in particular, with probability one when the distribution of the responses has a density. \suppsec{A.1} gives this argument for the full covariance matrix and its Vecchia representation.

In subsequent computations, we replace $R_A$ by $\widetilde R_A=B_A^{-1}D_AB_A^{-\top}$ and denote the resulting conditional integrated likelihood by $\widetilde q^I$. The remaining integral over the covariance parameters is evaluated using a Laplace approximation in the coordinates $\psi_A=(\log\phi_{A,1},\ldots,\log\phi_{A,d},\log g_A)^\top\in\mathbb R^{d+1}$. Write $d_\theta=d+1$ and $\theta_A=\tau_A(\psi_A)=\exp(\psi_A)$, with the exponential applied componentwise. Define
\[
\ell_A(\psi_A)
=\log\widetilde q^I\!\left(A;\tau_A(\psi_A)\right)
+\log\pi\!\left(\tau_A(\psi_A)\right)
+\log|\det J_{\tau_A}(\psi_A)|.
\]
With the Vecchia representation and this change of variables, $\Phi(A)$ is approximated by $\int_{\mathbb R^{d+1}}\exp\{\ell_A(\psi_A)\}\,d\psi_A$.
Let $\hat\psi_A$ be a local maximizer of $\ell_A$ inside the optimization bounds, and let
$G_{\psi,A}=-\nabla_{\psi_A}^2\ell_A(\hat\psi_A)$. When
$G_{\psi,A}$ is positive definite, the Laplace approximation gives
\begin{equation}
\log\Phi(A)\approx\log\widehat\Phi(A)
:= \ell_A(\hat\psi_A)
+\frac{d_\theta}{2}\log(2\pi)
-\frac{1}{2}\log|G_{\psi,A}|.
\label{eq:laplace_target}
\end{equation}

\vspace{-\parskip}

For the derivatives, fix the Vecchia ordering and conditioning sets of leaf $A$. Let $\widetilde q^I$, $M_A$, $h_A$, and $b_A$ denote the quantities of Proposition~\ref{prop:laplaceapprox} with $Q_A$ replaced by $\widetilde Q_A=B_A^\top D_A^{-1}B_A$, and write $V_A=M_A^{-1}$, $\bar\gamma_A=V_Ah_A$, and $e_A=\bfy(A)-H\bar\gamma_A$. Dots denote derivatives with respect to $\psi_A$: $\dot D_k=\partial D_A/\partial\psi_{A,k}$, $\dot{\widetilde Q}_k=\partial\widetilde Q_A/\partial\psi_{A,k}$, and $\dot M_k=H^\top\dot{\widetilde Q}_kH$ for $k=1,\ldots,d_\theta$, and $\ddot D_{k\ell}$, $\ddot{\widetilde Q}_{k\ell}$, and $\ddot M_{k\ell}$ are the corresponding second derivatives. Write $s_k=e_A^\top\dot{\widetilde Q}_ke_A$ and $s_{k\ell}=e_A^\top\ddot{\widetilde Q}_{k\ell}e_A-2(H^\top\dot{\widetilde Q}_ke_A)^\top V_A(H^\top\dot{\widetilde Q}_\ell e_A)$.

\begin{proposition}[Gradient]\label{prop:veechia_laplace}
If $B_A$ and $D_A$ are continuously differentiable in $\psi_A$, then for $k=1,\ldots,d_\theta$,
\[
\frac{\partial}{\partial\psi_{A,k}}\log\widetilde q^I
=-\frac12\operatorname{tr}(D_A^{-1}\dot D_k)
-\frac12\operatorname{tr}(V_A\dot M_k)
-\frac{a_A}{2b_A}\,s_k.
\]
\end{proposition}

\begin{proposition}[Hessian]\label{prop:nig_vecchia_hessian}
If $B_A$ and $D_A$ are twice continuously differentiable in $\psi_A$, then for $k,\ell=1,\ldots,d_\theta$,
\[
\begin{aligned}
\frac{\partial^2\log\widetilde q^I}{\partial\psi_{A,k}\,\partial\psi_{A,\ell}}
={}&\frac12\operatorname{tr}\bigl(D_A^{-1}\dot D_\ell D_A^{-1}\dot D_k-D_A^{-1}\ddot D_{k\ell}\bigr)\\
&+\frac12\operatorname{tr}\bigl(V_A\dot M_\ell V_A\dot M_k-V_A\ddot M_{k\ell}\bigr)
-a_A\Bigl(\frac{s_{k\ell}}{2b_A}-\frac{s_ks_\ell}{4b_A^2}\Bigr).
\end{aligned}
\]
\end{proposition}

The gradient and Hessian of $\ell_A$ are obtained by adding the corresponding derivatives of the covariance-parameter prior and the change-of-variables Jacobian. Proofs of Propositions~\ref{prop:veechia_laplace} and~\ref{prop:nig_vecchia_hessian} are given in Sections~A.2 and~A.3 of the Supplementary Material.

We locate $\hat\psi_A$ with L-BFGS-B, the limited-memory Broyden--Fletcher--Goldfarb--Shanno method with bound constraints \citep{Byrd1995}, followed by damped Newton iterations \citep[pp.~48--55]{NocedalWright2006}, with the analytic derivatives in Propositions~\ref{prop:veechia_laplace} and~\ref{prop:nig_vecchia_hessian}. Each node is optimized from a fixed sequence of starting values (from the fitted parameters of its parent in the storm-surge fits), and a node that is evaluated again reuses its stored fit, so each node has a single $\widehat\Phi(A)$ within a fit. \suppsec{B.1} gives the starting values, bounds, and convergence checks.

\subsection{Sequential Monte Carlo Method}\label{sc:smc}
Sequential Monte Carlo (SMC) approximates a sequence of probability distributions, called targets, by a population of $N$ weighted samples, called particles (\citealp[pp.~53--77]{Liu2001}; \citealp{ELEMENTOFSMC}). Successive targets are defined on increasingly larger spaces, and the last target is the distribution of interest. At each step, the particles are extended by draws from a proposal distribution, and their importance weights are multiplied by incremental weights that correct for drawing from the proposal instead of the target. When a few particles carry most of the weight, the population is resampled. For TVGP, we use the node-wise version of the top-down particle filter of \citet{Lakshminarayanan2013} for Bayesian decision trees. Each step makes the split-or-stop decision at one leaf, and the targets use the leaf marginal likelihood approximation $\widehat\Phi(A)$ of Equation~\eqref{eq:laplace_target}.

Each particle is a binary tree grown from the root, together with a first-in, first-out list of its active leaves, the leaves whose split-or-stop decision has not yet been made. At each step, the leaf at the front of the list is removed and either declared terminal or split into two children, which join the end of the list. This rule, also used by \citet{Lakshminarayanan2013} and \citet{Awaya2024}, processes leaves in breadth-first order, so all nodes at one depth are processed before any node at the next depth. Let $T_t^i$ denote the tree of particle $i$ after $t$ steps and $\mathcal U_t^i$ its list of active leaves; we drop the superscript $i$ when a statement holds for every particle. All particles start from the root-only tree, with the root as the only active leaf and with weight $1/N$. A particle whose list is empty is called complete and is left unchanged at later steps.

The target at step $t$ is a posterior distribution over partially grown trees whose prior is the Bayesian CART probability of the decisions made in the first $t$ steps \citep{Lakshminarayanan2013}. Its prior part multiplies the probabilities of Section~\ref{sc:CARTpriorIntro} over the processed nodes, and its likelihood part multiplies the leaf marginal likelihood approximations over the current leaves, active ones included. Up to a normalizing constant,
\begin{equation}
\widetilde\pi_t(T_t)=
\prod_{A\in\mathcal I_{T_t}}
\rho^\star(A)\lambda_{J(A)}(A)F_{L,A,J(A)}\!\left(\{L(A)\}\right)
\prod_{A\in\mathcal L_{T_t}\setminus\mathcal U_t}\{1-\rho^\star(A)\}
\prod_{A\in\mathcal L_{T_t}}\widehat\Phi(A),
\label{eq:partial_tree_target}
\end{equation}
where $\mathcal I_{T_t}$ and $\mathcal L_{T_t}$ are the internal nodes and the leaves of $T_t$, so that $\mathcal L_{T_t}\setminus\mathcal U_t$ contains its terminal leaves. The target at step $t$ is a function of $T_t$ alone. The fixed processing order determines $\mathcal U_t$ from $T_t$, and the same observations in a leaf always give the same $\widehat\Phi(A)$ within a fit (Section~\ref{sc:marginal_likelihood}). For a complete particle with tree $T$, Equation~\eqref{eq:partial_tree_target} reduces to $\widetilde\pi(T)=\pi(T)\prod_{A\in\mathcal L_T}\widehat\Phi(A)$,
which is proportional to the posterior distribution of $T$ with each leaf marginal likelihood replaced by its approximation $\widehat\Phi(A)$.

The proposal moves an incomplete particle by one decision at the first leaf $A$ of its list $\mathcal U_{t-1}$. By Equation~\eqref{eq:partial_tree_target}, stopping at $A$ multiplies the target by $1-\rho^\star(A)$, and splitting $A$ on axis $j\in\mathcal J(A)$ at $c\in\mathcal C_j(A)$ into children $A_L(j,c)$ and $A_R(j,c)$ multiplies it by $\rho^\star(A)\lambda_j(A)F_{L,A,j}(\{c\})\mathrm{BF}_A(j,c)$, where
\begin{equation}
\mathrm{BF}_A(j,c)=\frac{\widehat\Phi(A_L(j,c))\,\widehat\Phi(A_R(j,c))}{\widehat\Phi(A)}
\label{eq:one_step_ratio}
\end{equation}
is the Bayes factor of the split against keeping $A$ as one leaf. We call the factor by which a decision multiplies the target its target factor.

Drawing the decision from its prior, as in the bootstrap particle filter \citep{Gordon1993}, ignores these Bayes factors and can require many particles for an accurate approximation \citep{ELEMENTOFSMC}. We instead use the locally optimal proposal $\mathbb Q_A$, which draws each decision with probability proportional to its target factor and is the one-step optimal proposal of \citet{Lakshminarayanan2013}. For an active leaf $A$, let $Z_j(A)=\sum_{c\in\mathcal C_j(A)}F_{L,A,j}(\{c\})\mathrm{BF}_A(j,c)$ be the prior-weighted average Bayes factor of the splits on axis $j\in\mathcal J(A)$, and let $\kappa_A=1-\rho^\star(A)+\rho^\star(A)\sum_{j\in\mathcal J(A)}\lambda_j(A)Z_j(A)$ be the sum of the target factors over all decisions at $A$. Proposition~\ref{prop:posterior_sjl} gives $\mathbb Q_A$ in the order of the prior in Section~\ref{sc:CARTpriorIntro}, drawing the split indicator, then the axis, and then the cut point.

\begin{proposition}[Locally optimal proposal]\label{prop:posterior_sjl}
If $\mathcal J(A)=\varnothing$, then $\kappa_A=1$ and $\mathbb Q_A$ declares $A$ terminal with probability one. Otherwise, $\mathbb Q_A$ splits $A$ with probability $1-\{1-\rho^\star(A)\}/\kappa_A$; given a split, it draws the axis $j\in\mathcal J(A)$ with probability proportional to $\lambda_j(A)Z_j(A)$ and then the cut point $c\in\mathcal C_j(A)$ with probability proportional to $F_{L,A,j}(\{c\})\mathrm{BF}_A(j,c)$.
\end{proposition}

Under $\mathbb Q_A$, the odds of splitting are $\rho^\star(A)\{1-\rho^\star(A)\}^{-1}\sum_{j\in\mathcal J(A)}\lambda_j(A)Z_j(A)$, the prior odds of splitting times the prior-weighted average Bayes factor over all candidate splits. The leaf is therefore split more often than under the prior exactly when this average Bayes factor exceeds one. Because $\mathbb Q_A$ weighs all candidate splits at a leaf by their prior probabilities and Bayes factors, splits with large Bayes factors are proposed often, whereas RJ-MCMC proposes local changes to a complete tree and accepts or rejects them.

After the move, the weight of particle $i$ is updated as $W_t^i\propto W_{t-1}^iw_t^i$, and the weights are normalized to sum to one. The incremental weight is
\[
w_t^i=\frac{\widetilde\pi_t(T_t^i)}{\widetilde\pi_{t-1}(T_{t-1}^i)\,\mathbb Q_A(T_t^i\mid T_{t-1}^i)},
\]
where $\mathbb Q_A(T_t^i\mid T_{t-1}^i)$ is the probability under $\mathbb Q_A$ of the decision that turns $T_{t-1}^i$ into $T_t^i$. A complete particle has incremental weight one, because Equation~\eqref{eq:partial_tree_target} with an empty list equals $\widetilde\pi(T)$ at every later step.

\begin{proposition}[Incremental weight]\label{prop: increment_weights}
For a particle moved by $\mathbb Q_A$, the incremental weight equals $\kappa_A$.
\end{proposition}

Proofs of Propositions~\ref{prop:posterior_sjl} and~\ref{prop: increment_weights} are given in \suppsec{A.4}. The proposal is optimal in the sense that, given the particle before the move, the move adds no variability to the weights \citep{ELEMENTOFSMC}. Because $\kappa_A$ differs across particles, the normalized weights still become uneven as the steps proceed. After every step, the effective sample size (ESS), $1/\sum_{i=1}^N(W_t^i)^2$, measures how evenly the weights are spread. If $\mathrm{ESS}<N/2$, the population is resampled by drawing $N$ particles with replacement with probabilities $W_t^1,\ldots,W_t^N$. Each draw copies both the tree and its list of active leaves, and all weights are then reset to $1/N$.

Each split leaves at least the prescribed minimum number of training observations in both children, so the number of nodes in a tree is bounded and every particle is complete after a bounded number of steps. The algorithm stops when every list is empty, and the weighted particles then approximate the distribution proportional to $\widetilde\pi(T)$. Decisions are not revisited once made. Section~\ref{sec:sim1-sampler-comparison} compares the resulting approximation with RJ-MCMC. We denote the final particles by $\{(T_i,W_i)\}_{i=1}^N$, indexing particles by a subscript from here on, and use them for prediction in Section~\ref{sc:prediction}. \suppsec{E} gives the complete algorithm and its computational cost, which is dominated by evaluating $\widehat\Phi$ for both children of every candidate split; ordering the observations and finding the conditioning sets add further work.

\subsection{Prediction}\label{sc:prediction}
We use Bayesian model averaging \citep{Hoeting1999} over the final weighted particles \(\{(T_i,W_i)\}_{i=1}^N\), fixing each leaf's covariance parameters at \(\widehat\theta_A=\tau_A(\widehat\psi_A)\). For prediction inputs $X_{*A}$ in leaf $A$, with mean design matrix $H_{*A}$, predictions may target the latent function $H_{*A}\gamma_A+w_A(X_{*A})$ or future responses with independent $N(0,\sigma_A^2g_A)$ errors; $\bfy_{*A}$ denotes either target. Observation locations precede prediction locations in the extended Vecchia ordering \citep{KatzfussPrediction2020}, and the fitted observation factors are retained.

For a leaf $A$, let $C_{*A}$ denote the conditional regression matrix and $\sigma_A^2\Omega_{*A}$ the conditional covariance of $\bfy_{*A}$ given $\bfy(A)$, $\gamma_A$, and $\sigma_A^2$ under the extended Vecchia model. With $L_{*A}=H_{*A}-C_{*A}H$, and with $V_A$ and $\bar\gamma_A$ as defined before Proposition~\ref{prop:veechia_laplace} and evaluated at $\widehat\theta_A$, define $\mu_{*A}=C_{*A}\bfy(A)+L_{*A}\bar\gamma_A$ and $\Lambda_{*A}=\Omega_{*A}+L_{*A}V_AL_{*A}^{\top}$.
Integrating out $\gamma_A$ and $\sigma_A^2$ gives
\begin{equation}
\bfy_{*A}\mid\bfy(A),T,\widehat\theta_A
\sim t_{2a_A}\!\left(
\mu_{*A},\frac{b_A}{a_A}\Lambda_{*A}
\right),
\label{eq:nig_predictive_distribution}
\end{equation}
where the second argument is the scale matrix; the covariance is $b_A\Lambda_{*A}/(a_A-1)$ when $a_A>1$. The factor construction and derivation are given in \suppsec{A.5}.

For pointwise prediction, Equation~\eqref{eq:nig_predictive_distribution} is applied separately at each new input, with a conditioning set that contains only observed locations. For a new input $\bfx_*$ in leaf $A_i=A_i(\bfx_*)$ of tree $T_i$, the predictive density $\widehat p(y_*\mid\bfy)$ is the mixture of the Student-$t$ densities $t_{2a_{A_i}}(y_*;\mu_{*A_i},b_{A_i}\Lambda_{*A_i}/a_{A_i})$ with weights $W_i$, $i=1,\ldots,N$.
For the same conditioning set, the latent and noisy-response predictive distributions of each tree have the same mean, and the noisy-response variance is larger by $b_{A_i}\widehat g_{A_i}/(a_{A_i}-1)$ when $a_{A_i}>1$.

For joint prediction, a prediction location may also condition on preceding prediction locations within its leaf. Conditional on a tree, distinct leaves are independent. A joint draw selects one tree with probability $W_i$ and then draws a multivariate Student-$t$ vector for each leaf containing prediction locations. This yields dependent predictions for averages, maxima, and other functionals over sets of inputs; \suppsec{A.5} gives the joint mixture and its moments.

\section{Simulation Study} \label{sec: simulation study}

We use two simulation examples to answer three questions about TVGP. Section~\ref{sec:sim-examples} describes the examples, the TVGP settings, and the evaluation criteria shared by all experiments. Section~\ref{sec:sim1-sampler-comparison} asks how quickly and how reliably the SMC algorithm reaches accurate predictions, compared with RJ-MCMC targeting the same posterior over trees. Section~\ref{sec:ablation-main} measures how partitioning and the Vecchia approximation each affect predictive accuracy and computing time. Section~\ref{sc:sim1} compares TVGP with other nonstationary and scalable GP methods.

\subsection{Simulation Examples and Settings}\label{sec:sim-examples}

The first example tests whether a method can follow abrupt changes without adding spurious variation elsewhere. Let $\bfx=(x_1,x_2)^\top\in[0,1]^2$. The lines $x_1=0.30$ and $x_2=0.65$ divide the domain into four rectangles $\mathcal R_1=\{x_1\le 0.30,\ x_2>0.65\}$, $\mathcal R_2=\{x_1>0.30,\ x_2>0.65\}$, $\mathcal R_3=\{x_1\le 0.30,\ x_2\le 0.65\}$, and $\mathcal R_4=\{x_1>0.30,\ x_2\le 0.65\}$, and the response is $y(\bfx)=f(\bfx)+\varepsilon(\bfx)$ with $\varepsilon(\bfx)\sim\mathcal N\{0,\sigma^2(\bfx)\}$, where
\[
f(\bfx)=
\begin{cases}
0.8\,\sin(6\pi x_1)\cos(4\pi x_2), & \bfx\in\mathcal R_1,\\
1.5(x_1-0.5)+1.0(x_2-0.5), & \bfx\in\mathcal R_2,\\
3(x_1-0.2)^2+2(x_2-0.3)^2-0.2, & \bfx\in\mathcal R_3,\\
0.1, & \bfx\in\mathcal R_4,
\end{cases}
\qquad
\sigma(\bfx)=
\begin{cases}
0.10, & \bfx\in\mathcal R_1,\\
0.05, & \bfx\in\mathcal R_2,\\
0.12, & \bfx\in\mathcal R_3,\\
0.06, & \bfx\in\mathcal R_4.
\end{cases}
\]
The mean oscillates in $\mathcal R_1$, has a linear trend in $\mathcal R_2$, is a quadratic bowl in $\mathcal R_3$, and is constant in $\mathcal R_4$. A stationary GP uses the same ranges in all four rectangles. Ranges short enough to follow the oscillation then produce spurious wiggles in the constant region, and longer ranges blur the boundaries. The noise standard deviation also varies from 0.05 to 0.12 across the rectangles.

The second example has ten inputs and tests whether a method can follow changes in the oscillation frequency across $x_1$ while ignoring inputs that do not affect the response. Let $\operatorname{expit}(u)=\{1+\exp(-u)\}^{-1}$, define $s_1=\operatorname{expit}\{60(x_1-0.30)\}$, $s_2=\operatorname{expit}\{60(x_1-0.65)\}$, and $s_3=\operatorname{expit}\{60(x_2-0.55)\}$, and set $v_1=1-s_1$, $v_2=s_1-s_2$, and $v_3=s_2$. The latent mean is $f(\bfx)=b(\bfx)+\sum_{r=1}^3v_rh_r(\bfx)$, with
\[
\begin{aligned}
b(\bfx)={}&1.5(x_2-0.5)+2(x_3-0.5)^2
          +0.75\sin(2\pi x_4x_6),\\
h_1(\bfx)={}&0.6\sin(2\pi x_5),\\
h_2(\bfx)={}&\sin(6\pi x_5)+0.35\cos(2\pi x_2),\\
h_3(\bfx)={}&1.4\sin(16\pi x_5)
 +0.5\{(1-s_3)\cos(4\pi x_4)+s_3\cos(12\pi x_4)\},
\end{aligned}
\]
and the response is $y(\bfx)=f(\bfx)+\varepsilon$ with $\varepsilon\sim\mathcal N(0,0.05^2)$. The weights $v_1$, $v_2$, and $v_3$ switch from one component to the next near $x_1=0.30$ and $x_1=0.65$. Across these two transitions, the oscillation in $x_5$ goes from one cycle to three and then to eight cycles. Within the third component, $s_3$ changes the frequency in $x_4$ along $x_2$. Inputs $x_1,\ldots,x_6$ affect the response and $x_7,\ldots,x_{10}$ do not. Unlike in the first example, the noise variance is the same everywhere. We refer to the two examples as the two-dimensional and the 10-dimensional example.

\begin{table}[t]
\centering
\caption{Simulation experiments, with $n$ the training size and $n_{\mathrm{test}}$ the test size. Each dataset consists of a training design and a separate test design, both Latin hypercube designs \citep{McKay1979} of the type in the last column, generated with the R package \texttt{lhs} \citep{lhsPackage}.}
\label{tab:sim-designs}
\vspace{2pt}
\small\setlength{\tabcolsep}{4pt}
\begin{tabular}{@{}llrrll@{}}
\toprule
Section & Example & $n$ & $n_{\mathrm{test}}$ & Datasets & Design\\
\midrule
\ref{sec:sim1-sampler-comparison}, SMC versus RJ-MCMC & both & 5,000 & 2,000 & one each & random\\
\ref{sec:ablation-main}, partitioning and Vecchia & two-dimensional & 500--10,000 & 2,000 & ten per size & improved\\
\ref{sc:sim1}, other GP methods & two-dimensional & 10,000 & 2,000 & ten & improved\\
\ref{sc:sim1}, other GP methods & 10-dimensional & 2,000 & 400 & ten & maximin\\
\bottomrule
\end{tabular}
\end{table}

Table~\ref{tab:sim-designs} lists the four experiments. The training responses add independent normal errors to the latent function of the example. The sampler comparison runs each sampler ten times on one dataset per example, with a different random seed in each run. The other experiments fit each method once to each of ten datasets per training size, which differ in their designs and in their noise.

All TVGP fits use the product Mat\'ern correlation with smoothness $\nu=2.5$, conditioning sets of size $m=30$, ten candidate cut points per input, and a minimum leaf size of 30. They also use the tree prior of Section~\ref{sc:CARTpriorIntro} and the jointly robust and normal--inverse-gamma priors of Section~\ref{sc:marginal_likelihood}. SMC uses $N=5{,}000$ particles. In the two-dimensional example, the predictions change little with $N$ above 100 or with the hyperparameters of the jointly robust prior (Sections~C.2 and~C.5 of the Supplementary Material). In the sampler comparison, $g_A$ is bounded below by $10^{-6}$ for numerical stability, trees are limited to depth three, and every node takes its candidate cut points from one set of ten quantiles per input of all training inputs. The other experiments of this section use none of these settings. \suppsec{B} gives the settings of all methods and the computing resources of the experiments.

We assess predictions of the latent function at the test inputs by the root mean squared error (RMSE) of the predictive mean, the continuous ranked probability score \citep[CRPS;][]{Gneiting2007} of each method's pointwise predictive distribution, and the coverage and mean width of its 95\% intervals. Computing time is wall-clock time. SMC evaluates the candidate splits of a node in parallel on up to 16 cores, whereas an RJ-MCMC chain, including a TGP chain, evaluates one proposal at a time. We report means over runs or datasets in Sections~\ref{sec:sim1-sampler-comparison} and~\ref{sec:ablation-main} and medians over datasets in Section~\ref{sc:sim1}.

\subsection{SMC versus RJ-MCMC}\label{sec:sim1-sampler-comparison}

The two samplers use the same priors, leaf marginal likelihood approximation, and candidate splits. They therefore target the same approximate posterior over trees. Each RJ-MCMC chain starts from the tree with a single leaf and runs 50,000 iterations with a burn-in of 20,000, using the grow, prune, change, and swap proposals of \citet{Gramacy2012}. \suppsec{C.6} gives the complete sampler settings and trace plots of the 10-dimensional RJ-MCMC chains and of the TGP chains.

\begin{figure}[t]
    \centering
    \includegraphics[width=0.95\linewidth,keepaspectratio]{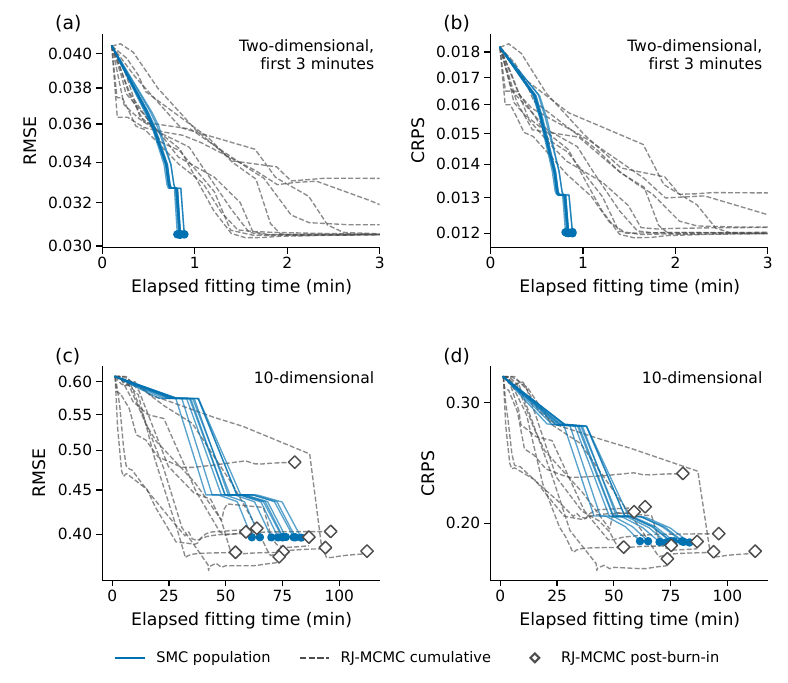}
    \caption{RMSE and CRPS versus fitting time in the two-dimensional example, (a)--(b), and the 10-dimensional example, (c)--(d), on logarithmic vertical axes; (a)--(b) show the first three minutes. Solid blue lines are ten SMC runs, each predicting from its current weighted population, with dots at completion; dashed gray lines are ten RJ-MCMC chains averaged from initialization, with open diamonds in (c)--(d) for the 30,000 retained draws. In the two-dimensional example, SMC reaches the final RJ-MCMC accuracy in about one tenth of the time RJ-MCMC takes for 50,000 iterations. In the 10-dimensional example, the SMC runs nearly coincide. The RJ-MCMC chains end at different trees, five of them with lower RMSE than SMC.}
    \label{fig:sim1-smc-rjmcmc-diagnostics}
\end{figure}

In the two-dimensional example, SMC reaches its final accuracy sooner than any RJ-MCMC chain (Figure~\ref{fig:sim1-smc-rjmcmc-diagnostics}(a)--(b)). The SMC runs finish after 0.8 minutes on average, compared with 8.7 minutes for the RJ-MCMC chains to complete 50,000 iterations. The final RMSE and CRPS of the two samplers are nearly identical. The cumulative RJ-MCMC estimates average the predictions over all iterations since initialization. In nine of the ten chains, these estimates come within 1\% of the mean final SMC RMSE after 1.4--5.1 minutes. The tenth chain stays above this level.

Although the RJ-MCMC chains predict about as well as SMC, they do not all reach the same trees. Two of the ten chains stay at trees with a lower sum of leaf log marginal likelihoods than the other eight chains and all ten SMC runs. The rank-normalized $\widehat R$ \citep{Vehtari2021} of this sum across the chains is 1.89. Values near one indicate that the chains agree.

In the 10-dimensional example, the two samplers take similar wall-clock time and are about equally accurate on average, but the SMC runs nearly coincide and the RJ-MCMC chains do not (Figure~\ref{fig:sim1-smc-rjmcmc-diagnostics}(c)--(d)). The mean fitting time is 74.0 minutes for SMC and 79.5 minutes for RJ-MCMC. The mean final RMSE is 0.397 for SMC and 0.400 for RJ-MCMC. All ten SMC runs end with the same highest-weight tree and nearly identical errors. The ten RJ-MCMC chains end at ten different trees, with RMSE after burn-in between 0.377 and 0.485. Five chains end with lower RMSE than SMC, at trees with higher approximate posterior probability than the common SMC tree. The cumulative estimates of two other chains never reach the SMC error level. Because SMC does not revise a split once it is made, runs that agree with each other can still all miss trees of higher posterior probability.

RJ-MCMC chains of TGP, the Bayesian treed GP of \citet{Gramacy2012} fitted with the \texttt{tgp} R package \citep{Gramacy2007TGP}, also remain at different trees. We run ten TGP chains of 50,000 iterations on each of the two datasets of this comparison. The rank-normalized $\widehat R$ of their log posterior is 2.43 in the two-dimensional example and 3.37 in the 10-dimensional example. RJ-MCMC chains for treed GP models can therefore remain at different trees for many iterations, in \texttt{tgp} as in our implementation.

\subsection{Contributions of Partitioning and the Vecchia Approximation}\label{sec:ablation-main}

To separate the contributions of partitioning and of the Vecchia approximation, the ablation uses the two-dimensional example and crosses two factors with two levels each: a single region or a tree learned with SMC, and the full covariance matrix or the Vecchia approximation in the likelihood. GP and Vecchia GP fit the leaf model of Section~\ref{sc:vecchiaGPIntro} to the whole domain as a single leaf, with the same priors and covariance estimation as TVGP. Treed GP and TVGP fit this model in each leaf of a learned tree. The two treed specifications share the tree prior, candidate splits, and SMC algorithm and differ only in the leaf likelihood. Unlike TGP, treed GP uses the priors and SMC algorithm of TVGP.
Each specification is fitted to the same ten datasets at each of $n=500$, $1{,}000$, and $2{,}000$. At $n=10{,}000$, only Vecchia GP and TVGP are fitted, because full-covariance likelihoods have cubic time and quadratic memory cost in $n$.

\begin{figure}[t]
\centering
\includegraphics[width=0.9\linewidth,keepaspectratio]{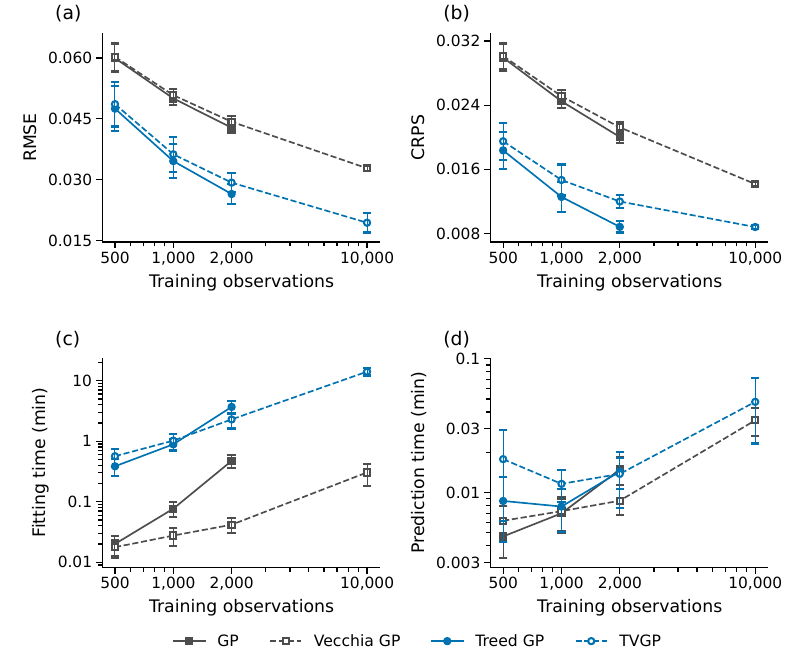}
\caption{Effects of partitioning and Vecchia likelihoods: (a) RMSE, (b) CRPS, (c) fitting time, and (d) prediction time. Points and bars give means $\pm$ one standard deviation over ten datasets at each training size. Gray denotes a single region and blue a learned tree; solid and dashed lines denote full-covariance and Vecchia likelihoods. Sample-size and time axes are logarithmic. Partitioning improves prediction under both likelihoods, and Vecchia likelihoods reduce the computing time of larger fits at a loss of accuracy in the treed fits, particularly in CRPS.}
\label{fig:ablation-components}
\end{figure}

Partitioning improves both RMSE and CRPS with either likelihood (Figure~\ref{fig:ablation-components}(a)--(b)). At $n=2{,}000$, it reduces mean RMSE by about 38\% with full-covariance likelihoods and by 34\% with Vecchia likelihoods. At $n=10{,}000$, the mean RMSE is 0.019 for TVGP and 0.033 for Vecchia GP. Separate regional fits help here because the surface contains distinct oscillatory, linear, quadratic, and constant components. The similar gains under the two likelihoods show that the benefit of partitioning does not depend on the leaf likelihood.

Within the treed fits, the Vecchia likelihood costs accuracy at every size where both likelihoods are fitted, particularly in CRPS. At $n=2{,}000$, the mean RMSE is 0.0293 for TVGP and 0.0265 for treed GP. The mean CRPS of TVGP is more than one third higher, 0.0120 against 0.0089. With conditioning sets of size 30, the Vecchia likelihood is exact only in leaves with at most 31 observations, which make up about 1\% of the particle-weighted leaves in the TVGP fits with $n=10{,}000$ (\suppsecp{C.1}).

The Vecchia approximation saves more fitting time as the sample size grows (Figure~\ref{fig:ablation-components}(c)). At $n=2{,}000$, it reduces fitting time by about 91\% for the single-region fit and by 38\% for the treed fit. At $n=500$ and $1{,}000$, TVGP takes longer than treed GP, because the sparse calculations save time only when the leaf covariance matrices are large enough to offset their overhead. Prediction takes less than 0.1 minutes for every specification (panel (d)), so fitting dominates the computing time. At $n=10{,}000$, Vecchia GP fits in 0.3 minutes on average and TVGP in 14 minutes, but Vecchia GP has much larger RMSE. A fast single-region fit therefore does not replace partitioning. \suppsec{D.2} adds coverage and the interval score.

\subsection{Comparison with Other GP Methods}\label{sc:sim1}

We compare TVGP with five methods, each representing a different approach to nonstationarity or to scalability. TGP \citep{Gramacy2012} is the classical Bayesian treed GP, with stationary GP leaves and RJ-MCMC over trees. Like TVGP, it uses a Mat\'ern correlation with $\nu=2.5$, but it differs from TVGP in its leaf likelihood, range parameters, priors, and software. The comparison with TGP therefore measures the combined effect of these differences. Scaled Vecchia (SVecchia; \citealp{KatzfussScaledVecchia}) is a global Vecchia GP with scaled inputs. It shows how well a scalable stationary model predicts without partitioning. The Vecchia-approximated deep Gaussian process (DGP VEC; \citealp{SauerDeepGP}) models nonstationarity by warping the inputs instead of partitioning them. The soft-partitioned Gaussian process (SPGP; \citealp{Luo2024}) uses soft instead of hard partitions. For each prediction input, the jump Gaussian process (JGP; \citealp{Park2022,Flowers2025}) fits a local GP that allows a jump between two regions. SPGP is used only in the two-dimensional example, because its implementation is restricted to two-dimensional inputs.

In the two-dimensional example, TVGP has the smallest prediction error (Figure~\ref{fig:tvgp_comparison_boxplots}(a)). Its median RMSE is 0.0180, compared with 0.0244 for TGP and 0.0286 for SPGP. TGP has the smallest median CRPS (0.0071), followed by TVGP (0.0087). Four of the six methods, TVGP, TGP, DGP VEC, and SPGP, have median coverage close to nominal, from 0.946 for TVGP and DGP VEC to 0.961 for TGP (Figure~\ref{fig:tvgp_comparison_boxplots}(b)). Among them, TGP has the narrowest intervals, consistent with its small CRPS. SVecchia and JGP undercover.

\begin{figure}[t]
    \centering
    \includegraphics[width=\linewidth,keepaspectratio]{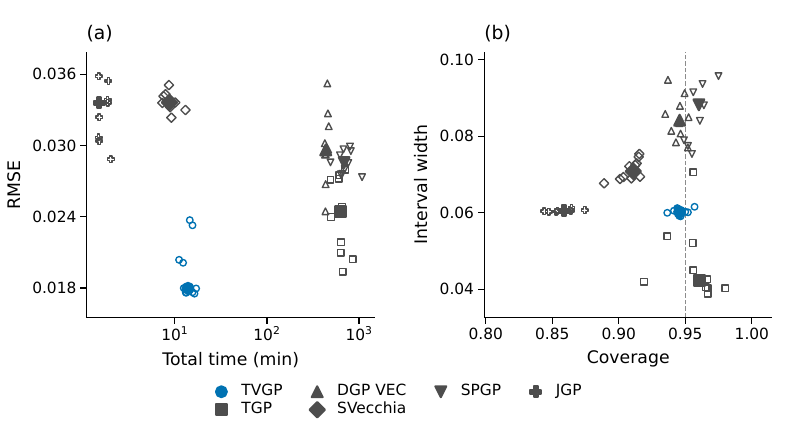}
    \caption{Prediction in the two-dimensional example: (a) RMSE versus total computing time on a logarithmic time axis and (b) 95\% interval width versus coverage, with a dashed line at 0.95. Small open symbols are the fits to the ten datasets, one per method and dataset (nine for JGP, whose fit to one dataset diverged); large filled symbols are medians. TVGP attains the lowest RMSE with near-nominal coverage, using less than one thirtieth of the time of TGP, DGP VEC, and SPGP. Mean scores of all methods appear in \suppsec{D.1}.}
    \label{fig:tvgp_comparison_boxplots}
\end{figure}

TVGP obtains this accuracy at a small fraction of the computing time of TGP, DGP VEC, and SPGP (Figure~\ref{fig:tvgp_comparison_boxplots}(a)). It takes a median of 14 minutes for fitting and prediction, compared with 7.2 hours for DGP VEC and more than ten hours for TGP and SPGP. SVecchia, at nine minutes, and JGP are faster but have much larger RMSE.

\begin{figure}[t]
    \centering
    \includegraphics[width=\linewidth,keepaspectratio]{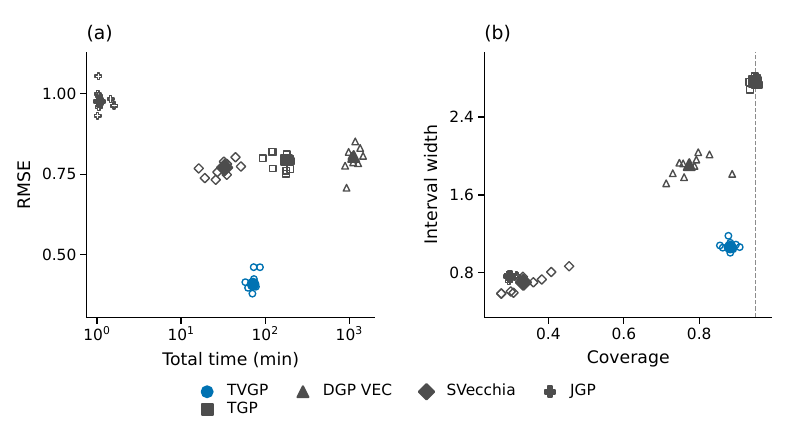}
    \caption{Prediction in the 10-dimensional example: (a) RMSE versus total computing time on a logarithmic time axis and (b) 95\% interval width versus coverage, with a dashed line at 0.95. Small open symbols are the fits to the ten datasets, one per method and dataset; large filled symbols are medians. TVGP has the lowest RMSE but undercovers. TGP has coverage closest to 0.95. Mean scores of all methods appear in \suppsec{D.1}.}
    \label{fig:ns10-axis-performance}
\end{figure}

In the 10-dimensional example\label{sec:synth-10d}, TVGP again has the smallest prediction error (Figure~\ref{fig:ns10-axis-performance}(a)), but its intervals undercover (Figure~\ref{fig:ns10-axis-performance}(b)). This comparison uses $n=2{,}000$ training inputs, because the computing times of TGP and DGP VEC increase rapidly with $n$. The median RMSE of TVGP is 0.408, compared with 0.77 for SVecchia, 0.79 for TGP, and 0.80 for DGP VEC. TVGP also has the smallest median CRPS of all methods (0.200). Its median coverage is 0.883. TGP has the coverage closest to the nominal level (0.948), with intervals more than twice as wide as those of TVGP. The coverage of the other methods ranges from 0.774 for DGP VEC down to 0.298 for JGP.

TVGP takes a median of 71.7 minutes on 16 cores for fitting and prediction, less than half the 3.0 hours of TGP and about one sixteenth of the 18.8 hours of DGP VEC. SVecchia, at 33 minutes, and JGP are faster but have much larger RMSE. A second version of this example tests boundaries that are not aligned with the coordinate axes. It replaces the two $x_1$ transitions by transitions in $(x_1+x_2)/2$ and keeps the designs, sizes, and number of datasets. TVGP again has the smallest mean RMSE and CRPS in this version (\suppsecp{C.4}), but not in a smooth two-dimensional example (\suppsecp{C.3}).

In summary, SMC finishes sooner in wall-clock time than RJ-MCMC in the two-dimensional example and about as soon in the 10-dimensional example, with similar final accuracy. Its runs agree with each other but can all miss trees of higher posterior probability. Partitioning is the main source of the improvement of TVGP over single-region fits. The Vecchia approximation reduces the fitting time of larger fits at a cost in accuracy, particularly in CRPS. Among the methods compared, TVGP has the smallest prediction error in both examples and takes less wall-clock time than TGP, DGP VEC, and SPGP. Its intervals, however, undercover in the 10-dimensional example.

\section{Storm-Surge Application} \label{sec: storm surge}
\subsection{Exploratory Data Analysis}

The storm-surge archive consists of 56,277 ADCIRC runs, one for each simulated storm scenario, and each scenario is characterized by seven input variables describing the intensity, size, pressure profile, motion, and landfall location of the storm. Table~\ref{tab:stormsurge_inputs} summarizes these inputs. The five storm parameters $\Delta P$, $R_p$, $V_f$, $\theta$, and $B$ take the values of a 3,500-point Latin hypercube design.
Then, for each of these 3,500 design points, the storm is simulated at landfall locations along the coastline with a spacing of 0.5 $R_p$ and a random initial position in the southwest direction. See \citet{Ma2022PPCokrig} for details of the computer model. These runs form an existing archive; the emulator predicts surge for additional storm scenarios without new ADCIRC runs, each of which is computationally expensive. We focus on the peak surge elevation, measured in meters, at a selected spatial location; extensions to many spatial locations are discussed in Section~\ref{sec: conclusion}. The response $y(\bfx)$ is the peak surge elevation at the input $\bfx=(\Delta P,R_p,V_f,\theta,B,\ell_x,\ell_y)^\top$.

\begin{table}[t]
\centering
\caption{Input variables for the storm-surge emulation. The storm parameters are $\Delta P$, $R_p$, $V_f$, $\theta$, and $B$, and the landfall location on the Southwest Florida coastline is given by its longitude $\ell_x$ and latitude $\ell_y$.}
\label{tab:stormsurge_inputs}
\vspace{2pt}
\small
\setlength{\tabcolsep}{4pt}
\begin{tabular}{lp{5.4cm}lp{2.6cm}}
\toprule
Symbol & Description & Unit & Range \\
\midrule
$\Delta P$ & Central pressure deficit of the storm & mb & [13, 93] \\
$R_p$ & Scale pressure radius & nautical miles & [8.11, 64.79] \\
$V_f$ & Storm's forward speed & m/s & [2, 15] \\
$\theta$ & Storm heading (clockwise from north) & degrees & [5, 120] \\
$B$ & Holland's $B$ parameter & unitless & [0.5, 1.7] \\
$\ell_x$ & Landfall longitude & degrees & [$-83.87$, $-81.06$] \\
$\ell_y$ & Landfall latitude & degrees & [25.00, 30.00] \\
\bottomrule
\end{tabular}
\end{table}

\begin{figure}[t]
\centering
\includegraphics[width=\linewidth,keepaspectratio]{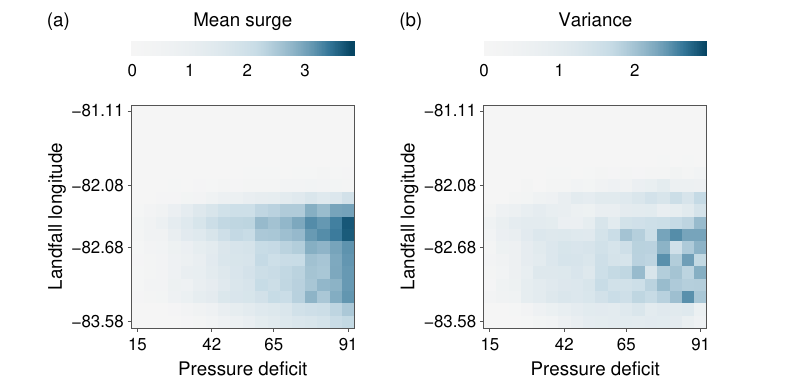}
\caption{Peak surge elevation at the selected location: (a) empirical mean in meters and (b) empirical variance in square meters over pressure-deficit and landfall-longitude bins. Pressure deficit is in millibars and longitude in degrees; tick labels give selected bin-center values. Within-bin variances include variation in the remaining storm inputs. The effect of pressure deficit depends strongly on landfall longitude.}
\label{fig:stormsurge_dp_landfall_heatmap}
\end{figure}

The selected location has longitude and latitude $(-82.052^\circ,\,26.928^\circ)$. About 62\% of the runs have zero peak surge elevation at this location. Figure~\ref{fig:stormsurge_dp_landfall_heatmap} shows the empirical mean (panel (a)) and variance (panel (b)) of peak surge elevation over pressure deficit and landfall longitude. In some landfall regions, the surge is zero or close to zero over much of the pressure-deficit range, whereas in others both the empirical mean and variance increase sharply with storm intensity.
Residual semivariograms, computed after removing additive trends in the storm inputs, also differ in level and shape across five bins of landfall locations along the coast (\suppsecp{C.8}). Together, these patterns motivate flexible regional modeling of the storm-surge response.

\subsection{Emulator Comparison at the Selected Location}

We use ten replications at each of two training sizes. In each replication, the 56,277 runs are randomly permuted. With 10,000 training runs, the first 10,000 runs form the training set and the next 1,000 the test set. With 50,649 training runs, the 5,628 runs after the first 10,000 form the test set and all remaining runs form the training set, so the smaller training and test sets are subsets of the larger ones. The zero responses are retained. For this application, TVGP uses $N=5{,}000$ particles, conditioning sets of size $m=70$, a minimum leaf size of 70, and at most 12 candidate splits per node. Because $\Phi(A)$ can be infinite under the jointly robust prior for leaves whose responses are all zero, TVGP uses independent inverse-gamma covariance priors here, under which $\Phi(A)$ is finite for every response vector (\suppsecp{B.1}). TVGP is compared with SVecchia, DGP VEC, JGP, and the local approximate Gaussian process \citep[laGP;][]{laGP}. TGP could not be fitted at these training sizes because of its computational cost, and the SPGP implementation is restricted to two inputs. Table~\ref{tab:stormsurge_results} summarizes pointwise predictions of the latent simulator response over these replications. Because the simulator is deterministic, coverage is computed for latent intervals, which omit the nugget variance and are narrower than noisy-response intervals (Section~\ref{sc:prediction}). \suppsec{B} gives the implementation details and computing resources.

\begin{table}[!htbp]
\centering
\caption{Latent predictive performance in the storm-surge application. TVGP has the smallest RMSE with 50,649 training runs and the highest coverage among the methods fitted ten times.}
\label{tab:stormsurge_results}
\vspace{2pt}
\small
\setlength{\tabcolsep}{5pt}
\begin{tabular}{lrrrrr}
\toprule
Method & RMSE & CRPS & Coverage & Mean width & Time (min) \\
\midrule
\multicolumn{6}{l}{\textit{Panel A: $n=10{,}000$, $n_{\mathrm{test}}=1{,}000$}} \\
TVGP & 0.312 & 0.129 & 0.864 & 0.504 & 31.8 \\
SVecchia & 0.296 & 0.150 & 0.783 & 0.570 & 24.7 \\
DGP VEC & 0.335 & 0.129 & 0.903 & 0.612 & 779.7 \\
JGP & 0.369 & 0.121 & 0.722 & 0.134 & 3.3 \\
laGP & 0.305 & 0.138 & 0.741 & 0.431 & 2.5 \\
\midrule
\multicolumn{6}{l}{\textit{Panel B: $n=50{,}649$, $n_{\mathrm{test}}=5{,}628$}} \\
TVGP & 0.254 & 0.094 & 0.881 & 0.343 & 247.1 \\
SVecchia & 0.258 & 0.118 & 0.792 & 0.399 & 39.0 \\
JGP & 0.279 & 0.070 & 0.674 & 0.063 & 17.6 \\
laGP & 0.259 & 0.100 & 0.748 & 0.285 & 13.7 \\
\bottomrule
\end{tabular}
\vspace{3pt}

\begin{minipage}{\linewidth}
\footnotesize
\textit{Note:} RMSE, CRPS, and mean interval width are in meters; coverage is for 95\% latent intervals. Time is total wall-clock time for fitting and prediction. Entries average ten replications, except DGP VEC, which is fitted once, to the first replication with 10,000 training runs, because a single fit takes about 13 hours.
\end{minipage}
\end{table}

TVGP has the smallest RMSE only with the larger training set. With 10,000 training runs, SVecchia (0.296) and laGP (0.305) have smaller RMSE than TVGP (0.312). With 50,649 training runs, TVGP has the smallest RMSE (0.254), closely followed by SVecchia (0.258) and laGP (0.259). JGP has the smallest CRPS, but its intervals are the narrowest and cover only 0.722 of the test runs with 10,000 training runs and 0.674 with 50,649, so its predictive distributions are sharp but overconfident. Among the methods fitted ten times, TVGP has the highest coverage at both training sizes, with narrower intervals than SVecchia. DGP VEC, fitted once, has higher coverage (0.903) with 10,000 training runs. Apart from DGP VEC, however, TVGP is the slowest method at both training sizes. With the larger training set, TVGP gives higher coverage than SVecchia and laGP but only a small reduction in RMSE, at about six and 18 times their wall-clock times.

\begin{figure}[t]
\centering
\includegraphics[width=0.95\linewidth,keepaspectratio]{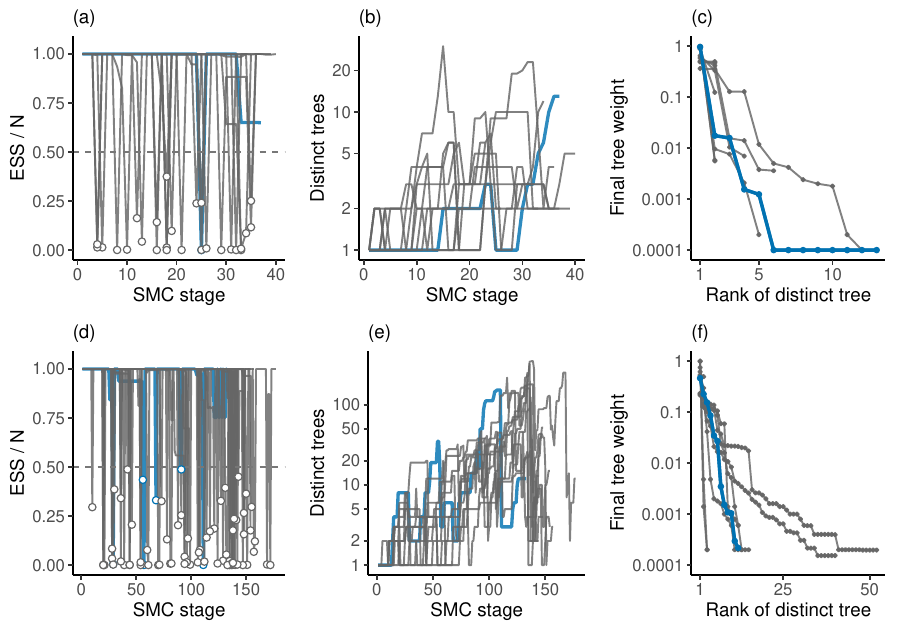}
\caption{Storm-surge SMC diagnostics for ten fits with 10,000 training runs, (a)--(c), and ten with 50,649, (d)--(f): (a), (d) effective sample size (ESS) before resampling divided by $N=5{,}000$; (b), (e) number of distinct trees after any resampling; (c), (f) final weights of the distinct trees in decreasing order. Open circles mark resampling, dashed lines mark the resampling threshold $\mathrm{ESS}/N=1/2$, and fit 1 is blue. Vertical axes in (b), (c), (e), and (f) are logarithmic; nine weights below $10^{-4}$ in (c) are drawn at $10^{-4}$. The final weight concentrates on a few trees.}
\label{fig:storm-current-smc}
\end{figure}

The intervals of TVGP undercover, and their coverage varies across the input space (\suppsecp{C.7}). At both training sizes, coverage is about 0.78 to 0.82 in the three westernmost landfall-longitude quintiles, compared with about 0.95 in the fourth and 1.00 in the easternmost. Coverage is also lower for positive responses than for zero responses. A second holdout scheme, with ten replications, holds out all landfall locations of one tenth of the storm design points, so that every test storm is new to the emulator. It shows the same regional pattern.

Leaves that contain both zero and positive training responses, which occur in all ten fits of the second scheme with about 50,650 training runs, account for part of the undercoverage. Positive-response test runs placed in such a leaf by most of 40 trees drawn from the final population have mean coverage 0.834, against 0.909 for the other positive responses. In a two-dimensional experiment with exact zeros over 30\% or 70\% of the domain, with ten datasets per scenario of 300 training and 1,000 test inputs from separate random Latin hypercube designs, TVGP has near-nominal overall coverage. Within 0.05 of the zero-response boundary, however, its coverage is 0.81 with 30\% zeros and 0.77 with 70\% zeros. The mixed leaves and this boundary experiment are consistent with incomplete separation of the zero-response region. Axis-aligned splits are restricted to a finite set of candidate cut points and a minimum leaf size, so some leaves contain responses from both sides of the zero-response boundary. One variance fitted to both kinds of responses can understate the variability of the positive ones. Positive responses outside such leaves also undercover, so the boundary explains only part of the shortfall.

Figure~\ref{fig:storm-current-smc} shows the SMC diagnostics for ten fits at each training size. Before many resampling events, marked by open circles in panels (a) and (d), $\mathrm{ESS}/N$ falls close to zero. One or a few trees then carry almost all of the weight, and resampling restores uniform weights by copying these trees. The number of distinct trees in panels (b) and (e) therefore drops at resampling, and equal weights after resampling do not indicate tree diversity. The final populations in panels (c) and (f) contain two to 13 distinct trees with $n=10{,}000$ and two to 52 with $n=50{,}649$. The largest tree weight exceeds 0.9 in four and two of the fits, respectively, so predictions average over few partitions.

\FloatBarrier
\section{Concluding Remarks} \label{sec: conclusion}
We have developed a Bayesian treed Vecchia Gaussian process model for nonstationary regression. TVGP combines recursive partitioning with regional Vecchia GP models, allowing the response surface and covariance structure to vary across input regions while retaining sparse leafwise likelihood calculations. An SMC algorithm with a locally optimal proposal learns the partition by using local marginal likelihoods to evaluate candidate splits, and posterior prediction averages over the final weighted tree population. Proper priors on the leaf covariance parameters regularize the repeated leafwise fits. Together, these components make Bayesian treed GP emulation feasible for nonstationary computer models with several inputs and tens of thousands of runs.

The intervals of TVGP undercover in the 10-dimensional example and the storm-surge application, and its final SMC weights often concentrate on a few trees. Hard partitions can introduce discontinuities across leaf boundaries, whereas many environmental response surfaces vary continuously over the input domain. Soft partitioning at the prediction stage \citep{Luo2024} and additive tree constructions \citep{Luo2022,Fox2012} offer two ways to smooth transitions between neighboring regions. More flexible partition geometries, including oblique hyperplane splits and minimum-spanning-tree partitions \citep{Luo2024}, could follow boundaries such as that of the zero-response region in Section~\ref{sec: storm surge} more closely. Local constructions from JGP models \citep{Park2022,Flowers2025} provide another direction for modeling behavior near partition boundaries. For small leaves with multimodal covariance posteriors, a sum of Laplace approximations centered at distinct local modes, or adaptive quadrature, could capture more of the posterior shape. The storm-surge simulator also produces outputs at many spatial locations. Extending TVGP with multivariate strategies \citep{Gu2016,Ma2022PPCokrig,Mccrimmon2026} would permit joint emulation across locations rather than separate fits at each site.

\bibliographystyle{apalike}
\bibliography{main}

\clearpage
\appendix
\setcounter{section}{0}
\setcounter{figure}{0}
\setcounter{table}{0}
\setcounter{equation}{0}
\setcounter{algorithm}{0}
\renewcommand{\thefigure}{S.\arabic{figure}}
\renewcommand{\thetable}{S.\arabic{table}}
\renewcommand{\theequation}{S.\arabic{equation}}
\renewcommand{\thealgorithm}{S.\arabic{algorithm}}
\def\theHfigure{S.\arabic{figure}}
\def\theHtable{S.\arabic{table}}
\def\theHequation{S.\arabic{equation}}
\def\theHalgorithm{S.\arabic{algorithm}}
\begin{center}
{\Large\bf Supplementary Material}
\end{center}
This Supplementary Material consists of five sections.
Section~\ref{sec: proof} contains the proofs of Propositions~3.1--3.5, shows that the leaf marginal likelihood is finite, and derives its Laplace approximation and the predictive distributions.
Section~\ref{sec: implementation} gives the settings of the treed Vecchia Gaussian process (TVGP) and of the competing methods in the simulation studies and the storm-surge application, and the computing resources.
Section~\ref{sec: additional results} provides additional numerical results, including diagnostics of the sequential Monte Carlo (SMC) population and comparisons of SMC with reversible-jump Markov chain Monte Carlo (RJ-MCMC).
Section~\ref{sec:complete-predictive} reports the complete predictive comparisons and additional calibration scores for the ablation study in Section~\ref{sec:ablation-main}.
Section~\ref{sec:smc_algorithm_complexity} presents the SMC algorithm and its computational complexity.

\section{Proofs} \label{sec: proof}
\subsection{Proof of Proposition 3.1}
\begin{proof}
Write $R_A=R_A(\theta_A)$, $Q_A=R_A^{-1}$, $P_0=V_0^{-1}$, $M_A=P_0+H^\top Q_AH$, $V_A=M_A^{-1}$, $h_A=P_0m_0+H^\top Q_A\bfy(A)$, and $\bar\gamma_A=V_Ah_A$. Under model~(4), in which $\gamma_A\mid\sigma_A^2\sim N_p(m_0,\sigma_A^2V_0)$ and $\sigma_A^2$ has the inverse-gamma distribution $\mathrm{IG}(a_0,b_0)$, completing the square in $\gamma_A$ gives
\begin{align}
&(\bfy(A)-H\gamma_A)^\top Q_A(\bfy(A)-H\gamma_A)+(\gamma_A-m_0)^\top P_0(\gamma_A-m_0)\notag\\
&\quad=(\gamma_A-\bar\gamma_A)^\top M_A(\gamma_A-\bar\gamma_A)+S_{0,A}^2,
\label{eq:nig_complete_square}
\end{align}
with $S_{0,A}^2=\bfy(A)^\top Q_A\bfy(A)+m_0^\top P_0m_0-h_A^\top V_Ah_A$ as in Section~3.1. Hence $\gamma_A\mid\sigma_A^2,\bfy(A),\theta_A\sim N_p(\bar\gamma_A,\sigma_A^2V_A)$, and integrating $\gamma_A$ against its normal prior gives
\begin{align*}
&\int p(\bfy(A)\mid\gamma_A,\sigma_A^2,\theta_A)\,p(\gamma_A\mid\sigma_A^2)\,d\gamma_A\\
&\quad=(2\pi\sigma_A^2)^{-n_A/2}|R_A|^{-1/2}|V_0|^{-1/2}|M_A|^{-1/2}\exp\!\left(-\frac{S_{0,A}^2}{2\sigma_A^2}\right).
\end{align*}
Multiplying by the $\mathrm{IG}(a_0,b_0)$ density of $\sigma_A^2$ shows that $\sigma_A^2\mid \bfy(A),\theta_A\sim\mathrm{IG}(a_A,b_A)$, with $a_A=a_0+n_A/2$ and $b_A=b_0+S_{0,A}^2/2$. The identity $\int_0^\infty x^{-a_A-1}\exp(-b_A/x)\,dx=\Gamma(a_A)b_A^{-a_A}$ then gives the integrated likelihood in Proposition~3.1.

\end{proof}

\paragraph{Finiteness of the leaf marginal likelihood.}
For fixed $\theta_A$, $q^I(A;\theta_A)$ of Proposition~3.1 is the marginal density of $\bfy(A)$ under the proper normal--inverse-gamma prior, so it integrates to one over $\mathbb R^{n_A}$. The same holds for $\widetilde q^I$, which uses the positive definite matrix $\widetilde R_A=B_A^{-1}D_AB_A^{-\top}$. Because the covariance prior is also proper, Tonelli's theorem gives $\int_{\mathbb R^{n_A}}\Phi(A)\,d\bfy(A)=1$ with either matrix, where $\Phi(A)=\int_{\Theta_A}q^I(A;\theta_A)\pi(\theta_A)\,d\theta_A$ as in Section~3.1. Hence $\Phi(A)<\infty$ for almost every $\bfy(A)$, and in particular with probability one when $\bfy(A)$ has a density.

Under the inverse-gamma covariance priors of Section~B.1, $\Phi(A)$ is finite for every $\bfy(A)$. Every observation contains an independent nugget component, so each Vecchia conditional variance satisfies $D_{A,i}\ge g_A$. Also, $M_A-P_0$ is positive semidefinite, giving $|V_0|^{-1/2}|M_A|^{-1/2}\le1$, and $b_A\ge b_0$. Proposition~3.1 therefore yields the uniform bound
\[
\widetilde q^I(A;\theta_A)
\le (2\pi)^{-n_A/2}g_A^{-n_A/2}
\frac{\Gamma(a_A)}{\Gamma(a_0)b_0^{n_A/2}}.
\]
For the full covariance matrix, $R_A-g_AI$ is positive semidefinite, so $|R_A|\ge g_A^{n_A}$ and the same bound holds for $q^I$. The bound does not depend on the ranges, and its integral against the inverse-gamma density of $g_A$ is finite, because this density decays faster than any power of $g_A$ as $g_A\to0$. Hence $\Phi(A)<\infty$ for every $\bfy(A)$ under these priors, including constant responses. When $g_A$ is bounded below by $10^{-6}$, as in Section~C.7, the same bound with $g_A$ replaced by $10^{-6}$ shows that $\Phi(A)$ is finite for every $\bfy(A)$ under the jointly robust prior renormalized on $g_A\ge10^{-6}$ (Section~B.1). Without the bound, $\Phi(A)$ can be infinite under the jointly robust prior. If $\bfy(A)=\mathbf 0$ and $m_0=0$, then $b_A=b_0$ for every $\theta_A$, and the factor $|R_A|^{-1/2}$ of $q^I$ grows faster than the prior density decays as the ranges grow and $g_A\to0$ together.

\paragraph{Laplace approximation.}
For computation, $R_A$ in Proposition~3.1 is replaced by $\widetilde R_A$, giving $\widetilde q^I$, and the covariance parameters are integrated in the coordinates $\psi_A=\log\theta_A\in\mathbb R^{d+1}$, taken componentwise, or with $\log(g_A-10^{-6})$ in place of $\log g_A$ when $g_A$ is bounded below (Section~B.1), with the integrand $\exp\{\ell_A(\psi_A)\}$ of Section~3.1, whose prior and Jacobian follow Section~B.1 when $g_A$ is bounded below. This integral has no closed form in general. At the local mode $\hat\psi_A$ inside the optimization bounds, selected as in Section~B.1, with $G_{\psi,A}=-\nabla_{\psi_A}^2\ell_A(\hat\psi_A)$ positive definite, integrating the second-order Taylor expansion of $\ell_A$ about $\hat\psi_A$ gives Equation~(5).

\subsection{Proof of Proposition 3.2}
\label{sc:proof_prop_2}
\begin{proof}
Drop the leaf index from the Vecchia quantities and from $\psi$, and condition on a fixed Vecchia ordering and fixed conditioning sets $N(i)$. Let $R(\psi)$ denote the ordered correlation matrix, including the nugget term. The Vecchia factors satisfy
\[
B_{ii}=1,\qquad
B_{i,N(i)}=-R_{i,N(i)}R_{N(i),N(i)}^{-1},\qquad
D_{ii}=R_{ii}+B_{i,N(i)}R_{N(i),i},
\]
with all other entries of $B$ equal to zero. Hence
\[
\widetilde Q=B^\top D^{-1}B,
\qquad
\widetilde R=B^{-1}DB^{-\top}.
\]
Because $B$ is unit lower triangular, $|\widetilde R|=|D|$.

Set $P_0=V_0^{-1}$,
$M_A=P_0+H^\top\widetilde QH$, $V_A=M_A^{-1}$, and
\[
\bar\gamma_A
=V_A(P_0m_0+H^\top\widetilde Q\bfy(A)),
\qquad
e_A=\bfy(A)-H\bar\gamma_A,
\qquad
\delta=\bar\gamma_A-m_0.
\]
Applying Proposition~3.1 with $R$ replaced by $\widetilde R$ gives
\begin{equation}
\log\widetilde q^I(\psi)
=C-\frac12\log|D|-\frac12\log|M_A|
-a_A\log b_A,
\label{eq:nig_vecchia_log_integrated}
\end{equation}
where $C$ does not depend on $\psi$,
$a_A=a_0+n_A/2$, and
\[
b_A=b_0+\frac12S_{0,A}^2,
\qquad
S_{0,A}^2
=e_A^\top\widetilde Qe_A+\delta^\top P_0\delta.
\]

\vspace{-\parskip}
The posterior mean $\bar\gamma_A$ satisfies
\[
M_A\bar\gamma_A=P_0m_0+H^\top\widetilde Q\bfy(A).
\]
Rearranging this equation yields the penalized normal equation
\begin{equation}
H^\top\widetilde Qe_A=P_0\delta.
\label{eq:nig_penalized_normal}
\end{equation}
For $k=1,\ldots,d_\theta$, where $d_\theta=d+1$ as in Section~3.1, a dot denotes differentiation with respect to $\psi_k$, as in $\dot{\widetilde Q}_k=\partial\widetilde Q/\partial\psi_k$, and differentiation of
$S_{0,A}^2$ gives
\begin{align}
\frac{\partial S_{0,A}^2}{\partial\psi_k}
&=
2e_A^\top\widetilde Q\frac{\partial e_A}{\partial\psi_k}
+e_A^\top\dot{\widetilde Q}_ke_A
+2\delta^\top P_0\frac{\partial\delta}{\partial\psi_k}\notag\\
&=
-2e_A^\top\widetilde QH
  \frac{\partial\bar\gamma_A}{\partial\psi_k}
+e_A^\top\dot{\widetilde Q}_ke_A
+2\delta^\top P_0
  \frac{\partial\bar\gamma_A}{\partial\psi_k}\notag\\
&=e_A^\top\dot{\widetilde Q}_ke_A,
\label{eq:nig_s0_derivative}
\end{align}
where the last equality follows from
Equation~\eqref{eq:nig_penalized_normal}. Therefore,
\[
\frac{\partial b_A}{\partial\psi_k}
=\frac12e_A^\top\dot{\widetilde Q}_ke_A.
\]
The remaining derivatives in
Equation~\eqref{eq:nig_vecchia_log_integrated} are
\[
\frac{\partial}{\partial\psi_k}\log|D|
=\operatorname{tr}(D^{-1}\dot D_k),
\qquad
\frac{\partial}{\partial\psi_k}\log|M_A|
=\operatorname{tr}(V_A\dot M_k),
\]
where $\dot M_k=H^\top\dot{\widetilde Q}_kH$.
The derivatives of the Vecchia factors follow from those of $R$. For each $i$, with $N=N(i)$ and the index $k$ omitted from $\dot B$, $\dot D$, and $\dot R$,
\[
\dot B_{i,N}=-(\dot R_{i,N}+B_{i,N}\dot R_{N,N})R_{N,N}^{-1},
\qquad
\dot D_{ii}=(1,B_{i,N})\,\dot R_{\{i\}\cup N}\,(1,B_{i,N})^\top,
\]
where $\dot R_{\{i\}\cup N}$ is the derivative of the correlation matrix of $\{i\}\cup N(i)$ with $i$ ordered first. Then
\[
\dot{\widetilde Q}_k=\dot B^\top D^{-1}B+B^\top D^{-1}\dot B-B^\top D^{-1}\dot DD^{-1}B.
\]
Substitution into Equation~\eqref{eq:nig_vecchia_log_integrated}
gives the gradient stated in Proposition~3.2. Differentiating $\dot B$, $\dot D$, and $\dot{\widetilde Q}_k$ once more with respect to $\psi_\ell$ gives $\ddot B_{k\ell}$, $\ddot D_{k\ell}$, and $\ddot{\widetilde Q}_{k\ell}$, which Section~A.3 uses.

\end{proof}

\subsection{Proof of Proposition 3.3}
\label{sc:proof_prop_3}
\begin{proof}
Retain the notation of Section~A.2 and write
\[
\widetilde f(\psi)=\log\widetilde q^I(\psi)
=C-\frac12\log|D|-\frac12\log|M_A|-a_A\log b_A,
\]
where $C$ does not depend on $\psi$. Differentiating
$M_A\bar\gamma_A=P_0m_0+H^\top\widetilde Q\bfy(A)$ with respect to $\psi_\ell$ gives
\[
\frac{\partial\bar\gamma_A}{\partial\psi_\ell}
=V_A H^\top\dot{\widetilde Q}_\ell e_A
=V_A z_\ell,
\qquad
\frac{\partial e_A}{\partial\psi_\ell}
=-HV_A z_\ell,
\]
where $z_k=H^\top\dot{\widetilde Q}_ke_A$ for $k=1,\ldots,d_\theta$.
Since Equation~\eqref{eq:nig_s0_derivative} gives
$s_k=\partial S_{0,A}^2/\partial\psi_k
=e_A^\top\dot{\widetilde Q}_ke_A$, a second differentiation yields
\begin{align*}
s_{k\ell}
&=\frac{\partial s_k}{\partial\psi_\ell}\\*
&=e_A^\top\ddot{\widetilde Q}_{k\ell}e_A
+2\left(\frac{\partial e_A}{\partial\psi_\ell}\right)^\top
\dot{\widetilde Q}_ke_A\\*
&=e_A^\top\ddot{\widetilde Q}_{k\ell}e_A
-2z_k^\top V_A z_\ell.
\end{align*}

\vspace{-\parskip}
The inverse-matrix derivatives are
\[
\frac{\partial D^{-1}}{\partial\psi_\ell}
=-D^{-1}\dot D_\ell D^{-1},
\qquad
\frac{\partial V_A}{\partial\psi_\ell}
=-V_A\dot M_\ell V_A.
\]
Therefore,
\begin{align}
\frac{\partial}{\partial\psi_\ell}
\left\{-\frac12\operatorname{tr}(D^{-1}\dot D_k)\right\}
&=\frac12\operatorname{tr}\!\left(
D^{-1}\dot D_\ell D^{-1}\dot D_k
-D^{-1}\ddot D_{k\ell}\right),
\label{eq:nig_hess_logD}\\
\frac{\partial}{\partial\psi_\ell}
\left\{-\frac12\operatorname{tr}(V_A\dot M_k)\right\}
&=\frac12\operatorname{tr}\!\left(
V_A\dot M_\ell V_A\dot M_k
-V_A\ddot M_{k\ell}\right).
\label{eq:nig_hess_logM}
\end{align}
Because
$\partial b_A/\partial\psi_\ell=s_\ell/2$,
\begin{align}
\frac{\partial}{\partial\psi_\ell}
\left(-\frac{a_A}{2b_A}s_k\right)
&=-a_A\left\{
\frac{s_{k\ell}}{2b_A}
-\frac{s_ks_\ell}{4b_A^2}
\right\}.
\label{eq:nig_hess_residual}
\end{align}
Adding Equations~\eqref{eq:nig_hess_logD}, \eqref{eq:nig_hess_logM},
and~\eqref{eq:nig_hess_residual}
gives the Hessian stated in Proposition~3.3.
\end{proof}

All derivatives in Propositions~3.2--3.3 are taken with respect to the
coordinates $\psi$ used in the covariance-parameter integral. In particular,
if $\theta_r=\exp(\psi_r)$, then
\begin{equation}
\frac{\partial R}{\partial\psi_r}
=\theta_r\frac{\partial R}{\partial\theta_r},
\qquad
\frac{\partial^2R}{\partial\psi_r\partial\psi_s}
=\theta_r\theta_s
\frac{\partial^2R}{\partial\theta_r\partial\theta_s}
+\boldsymbol{1}\{r=s\}\theta_r
\frac{\partial R}{\partial\theta_r}.
\label{eq:nig_log_coordinate_second_derivative}
\end{equation}
For $v=\log g$ and $R=R_0+gI$, where $R_0$ is the correlation matrix without the nugget, the chain rule
gives $\partial R/\partial v=\partial^2R/\partial v^2=gI$; mixed range--nugget derivatives are zero.

\subsection{Proof of Propositions 3.4 and 3.5}
\begin{proof}
Consider a particle with tree \(T_{t-1}\) and list of active leaves \(\mathcal U_{t-1}\), and let \(A\) be the first leaf in \(\mathcal U_{t-1}\). Before \(A\) is processed, it contributes \(\widehat\Phi(A)\) to the likelihood part of the target in Equation~(6) and no prior factor. If growth stops at \(A\), then \(A\) becomes a terminal leaf, the leaves of the tree are unchanged, and the target acquires the factor \(1-\rho^\star(A)\). Hence
\[
\frac{\widetilde\pi_t(T_t)}{\widetilde\pi_{t-1}(T_{t-1})}=1-\rho^\star(A).
\]
If \(A\) is split on axis \(j\) at \(c\), then \(A\) becomes an internal node, its children \(A_L(j,c)\) and \(A_R(j,c)\) replace it among the leaves and are appended to the list of active leaves, and the target acquires the split-prior factor \(\rho^\star(A)\lambda_j(A)F_{L,A,j}(\{c\})\), where \(F_{L,A,j}\) is the cut-point prior of Section~2.2. Therefore
\begin{align*}
\frac{\widetilde\pi_t(T_t)}{\widetilde\pi_{t-1}(T_{t-1})}
&=\rho^\star(A)\lambda_j(A)F_{L,A,j}(\{c\})
\frac{\widehat\Phi(A_L(j,c))\widehat\Phi(A_R(j,c))}{\widehat\Phi(A)}\\
&=\rho^\star(A)\lambda_j(A)F_{L,A,j}(\{c\})\mathrm{BF}_A(j,c).
\end{align*}
Write \(\eta_A(0)=1-\rho^\star(A)\) and \(\eta_A(j,c)=\rho^\star(A)\lambda_j(A)F_{L,A,j}(\{c\})\mathrm{BF}_A(j,c)\) for these ratios, and let \(\Delta_A=\{0\}\cup\{(j,c):j\in\mathcal J(A),\,c\in\mathcal C_j(A)\}\) be the set of admissible decisions. The locally optimal proposal draws a decision \(\delta\in\Delta_A\) with probability proportional to its target ratio,
\[
\mathbb Q_A(\delta)=\frac{\eta_A(\delta)}{\sum_{\delta'\in\Delta_A}\eta_A(\delta')}.
\]
By the definition of \(Z_j(A)\) in Section~3.2, the normalizing constant is
\[
\sum_{\delta'\in\Delta_A}\eta_A(\delta')
=1-\rho^\star(A)+\rho^\star(A)\sum_{j\in\mathcal J(A)}\lambda_j(A)Z_j(A)=\kappa_A.
\]
Summing \(\eta_A(j,c)/\kappa_A\) over all split decisions gives \(\mathbb Q_A\{S(A)=1\}=1-\{1-\rho^\star(A)\}/\kappa_A\). For a fixed axis \(j\), the sum over \(c\in\mathcal C_j(A)\) equals \(\rho^\star(A)\lambda_j(A)Z_j(A)/\kappa_A\); dividing this sum by \(\mathbb Q_A\{S(A)=1\}\) gives the axis probability, and dividing \(\eta_A(j,c)/\kappa_A\) by the sum gives the cut-point probability in Proposition~3.4.

For a sampled decision \(\delta_t\), the incremental importance weight is
\[
w_t=\frac{\widetilde\pi_t(T_t)}{\widetilde\pi_{t-1}(T_{t-1})\,\mathbb Q_A(\delta_t)}
=\frac{\eta_A(\delta_t)}{\eta_A(\delta_t)/\kappa_A}=\kappa_A,
\]
which proves Proposition~3.5. If \(\mathcal J(A)=\varnothing\), then \(\rho^\star(A)=0\), \(\Delta_A=\{0\}\), and \(\mathbb Q_A(0)=\kappa_A=1\). For a complete particle, the transition is the identity and the target is unchanged, so its incremental weight is one. Once a particle is complete, every internal node contributes its split-prior factor and every leaf contributes \(1-\rho^\star(A)\); their product is the Bayesian classification and regression tree (CART) prior of Section~2.2.
\end{proof}

\subsection{Derivation of the Predictive Distributions}
Fix leaf $A$ of a tree and its covariance parameters at the mode $\widehat\theta_A=\tau_A(\widehat\psi_A)$, with $\tau_A$ as in Section~3.1, and let $\delta_*=0$ for latent-function prediction and $\delta_*=1$ for noisy-response prediction. With $K_{AA}$, $K_{A*}=K_{*A}^\top$, and $K_{**}$ the correlation matrices among the observed and prediction locations of the leaf at $\widehat\theta_A$, the prediction rows of the Vecchia factors are formed from
\[
R_A^{\mathrm{ext}}(\delta_*)=
\begin{pmatrix}
K_{AA}+g_A I_{n_A} & K_{A*}\\
K_{*A} & K_{**}+\delta_*g_A I_{n_{*A}}
\end{pmatrix}.
\]
For a prediction row $i$ with predecessor set $N_i$, $B_{i,N_i}=-R^{\mathrm{ext}}_{i,N_i}(R^{\mathrm{ext}}_{N_i,N_i})^{-1}$ and $D_i=R^{\mathrm{ext}}_{ii}+B_{i,N_i}R^{\mathrm{ext}}_{N_i,i}$, whereas the fitted observation factors $B_A$ and $D_A$ are kept. Under an observations-first extended Vecchia ordering,
\begin{equation}
\begin{pmatrix}B_A&0\\ B_{*A}&B_*\end{pmatrix}
\begin{pmatrix}\bfy(A)-H\gamma_A\\ \bfy_{*A}-H_{*A}\gamma_A\end{pmatrix}
=\begin{pmatrix}\boldsymbol\varepsilon_A\\ \boldsymbol\varepsilon_*\end{pmatrix},
\qquad
\begin{pmatrix}\boldsymbol\varepsilon_A\\ \boldsymbol\varepsilon_*\end{pmatrix}
\sim\mathcal N\!\left[0,\sigma_A^2\begin{pmatrix}D_A&0\\0&D_*\end{pmatrix}\right],
\label{eq:extended_vecchia_prediction_system}
\end{equation}
so that $\bfy_{*A}=C_{*A}\bfy(A)+L_{*A}\gamma_A+B_*^{-1}\boldsymbol\varepsilon_*$ with $C_{*A}=-B_*^{-1}B_{*A}$ and $L_{*A}=H_{*A}-C_{*A}H$. Because $\boldsymbol\varepsilon_*$ is independent of $\boldsymbol\varepsilon_A$,
\[
\bfy_{*A}\mid\gamma_A,\sigma_A^2,\bfy(A),\widehat\theta_A
\sim\mathcal N\!\left(C_{*A}\bfy(A)+L_{*A}\gamma_A,\sigma_A^2\Omega_{*A}\right),
\qquad \Omega_{*A}=B_*^{-1}D_*B_*^{-\top}.
\]
By Section~A.1, with $\widetilde Q_A=B_A^\top D_A^{-1}B_A$ in place of $Q_A$,
\[
\gamma_A\mid\sigma_A^2,\bfy(A),\widehat\theta_A\sim\mathcal N(\bar\gamma_A,\sigma_A^2V_A),
\qquad
\sigma_A^2\mid\bfy(A),\widehat\theta_A\sim\mathrm{IG}(a_A,b_A).
\]
Integrating $\gamma_A$ gives $\bfy_{*A}\mid\sigma_A^2,\bfy(A),\widehat\theta_A\sim\mathcal N(\mu_{*A},\sigma_A^2\Lambda_{*A})$, with $\mu_{*A}=C_{*A}\bfy(A)+L_{*A}\bar\gamma_A$ and $\Lambda_{*A}=\Omega_{*A}+L_{*A}V_AL_{*A}^\top$, and integrating $\sigma_A^2$ gives
\[
\bfy_{*A}\mid\bfy(A),T,\widehat\theta_A
\sim t_{2a_A}\!\left(\mu_{*A},\frac{b_A}{a_A}\Lambda_{*A}\right),
\]
whose covariance under the scale-matrix parameterization is $b_A\Lambda_{*A}/(a_A-1)$ when $a_A>1$.

For pointwise prediction, each prediction location is appended separately, so $B_*=1$. For joint prediction, preceding prediction locations from the same leaf may enter the conditioning sets, and the joint distribution for tree $T_i$ factors over the leaves containing prediction locations,
\[
p_i(\bfy_*\mid\bfy)=\prod_{A\in\mathcal L(T_i):\,n_{*A}>0}t_{2a_A}\!\left(\bfy_{*A};\mu_{*A},\frac{b_A}{a_A}\Lambda_{*A}\right),
\]
with Bayesian model average $\widehat p_{\mathrm{joint}}(\bfy_*\mid\bfy)=\sum_{i=1}^NW_i\,p_i(\bfy_*\mid\bfy)$. A joint draw selects tree $T_i$ with probability $W_i$ and then draws one multivariate Student-$t$ vector for each leaf containing prediction locations. If $a_A>1$ in every contributing leaf and $p_i$ has mean $\mu_i$ and covariance $\Sigma_i$, the model-averaged mean is $\bar\mu_*=\sum_iW_i\mu_i$ and the covariance is $\sum_iW_i\Sigma_i+\sum_iW_i(\mu_i-\bar\mu_*)(\mu_i-\bar\mu_*)^\top$.

Within each tree, the noisy-response and latent predictive variances in leaf $A$ differ by $b_A\widehat g_A/(a_A-1)$ (Section~3.3). To see this, fix the same conditioning set of observed locations for latent and noisy-response prediction at a new point. In both cases the observed locations have covariance $K_{NN}+g_AI$ and cross-covariance $K_{*N}$ with the new point, so the regression coefficients, the predictive mean, and the term for uncertainty in the mean coefficients are identical. Changing the new-point diagonal from $1$ to $1+g_A$ increases $D_*$, and hence $\Lambda_{*A}$, by $g_A$. Multiplying by the posterior expectation $\mathbb E(\sigma_A^2\mid\bfy(A),\widehat\theta_A)=b_A/(a_A-1)$ gives the stated difference.

\section{Implementation Details} \label{sec: implementation}
In the simulation studies, unless stated otherwise, TVGP uses $N=5{,}000$ particles, conditioning sets of size $m=30$, Mat\'ern smoothness $\nu=2.5$, minimum leaf size 30, ten candidate cut points per input at empirical quantiles of the node's training inputs, tree-prior parameters $\alpha=0.95$ and $\beta=0.5$, normal--inverse-gamma hyperparameters $(m_0,V_0,a_0,b_0)=(0,100,2,1)$, and the jointly robust covariance prior with $a=0.2$ and $b=1$. In the storm-surge comparison, TVGP uses $N=5{,}000$ particles, conditioning sets of size $m=70$, $\nu=2.5$, minimum leaf size 70, the same tree-prior parameters, and the inverse-gamma covariance priors of Section~B.1, and each input is scaled using its training-set minimum and maximum, with the same transformation applied to the test inputs for every method. The CRPS and the 95\% intervals of TVGP are computed from 300 draws from its predictive distribution at each test input, 1,000 in the sampler comparison.

\subsection{TVGP Optimization and Priors}
\paragraph{Covariance-parameter optimization.}
We maximize $\ell_A(\psi_A)$ in the log coordinates $u_{A,j}=\log\phi_{A,j}$ and $v_A=\log g_A$ by the limited-memory Broyden--Fletcher--Goldfarb--Shanno method with bound constraints (L-BFGS-B) of \citet{Byrd1995}, followed by damped Newton iterations \citep[pp.~48--55]{NocedalWright2006}, with analytic gradient and Hessian. With $r_{A,j}$ the span of input $j$ in node $A$, the first three attempts start from $(\phi_{A,j}/r_{A,j},g_A)=(0.20,0.05)$, $(0.50,0.10)$, and $(0.08,0.01)$ and bound $\phi_{A,j}/r_{A,j}$ to $[10^{-3},10]$ and $g_A$ to $[10^{-8},10]$. Later attempts restart from the first two with wider bounds, in the widest case $[10^{-4},10^4]$ and $[10^{-14},100]$, which do not change the prior. The first attempt that converges with $\|\nabla\ell_A(\widehat\psi_A)\|_\infty\le10^{-3}$, lies more than $10^{-6}$ from every numerical bound, and has positive-definite $G_{\psi,A}$ is retained. If no attempt passes these checks, the run stops with an error rather than discarding the split. The starting values depend only on the node's data and fixed optimizer settings, except in the storm-surge fits. Within each leaf, the Vecchia ordering is the global training-input ordering restricted to the leaf.

\paragraph{Screened starting values.}
In the sampler comparison of Section~4.2, optimization first screens $d+3$ starting values, the three above and, for each input $k$, a short-range start with $\phi_{A,k}/r_{A,k}=0.03$, $\phi_{A,j}/r_{A,j}=0.30$ for $j\neq k$, and $g_A=0.01$, ranks them by $\ell_A$ at the starting value, and optimizes the three highest-ranked starts in that order, with range bounds from $10^{-3}$ to $10^3$ times the span. The first of these fits that passes the checks above is retained, and the attempts above follow if none passes. A short-range start in one input lets the optimizer reach modes in which that input carries high-frequency variation. In the 10-dimensional example, the default starts alone can lead to a mode with a long range in $x_5$ and a large nugget, which treats the oscillation in $x_5$ as noise.

\paragraph{Covariance prior.}
We use the jointly robust (JR) prior in Equation~(3.2) of \citet{Gu2019}. For \(\zeta_{A,j}=1/\phi_{A,j}\) and
\[
t_A=\sum_{j=1}^d C_{A,j}\zeta_{A,j}+g_A,
\qquad
C_{A,j}=n_A^{-1/d}\{\max_i x_{A,i,j}-\min_i x_{A,i,j}\},
\]
assume positive coordinate spans, so that $C_{A,j}>0$. The normalized prior is
\[
\pi_{\mathrm{JR}}(\boldsymbol\zeta_A,g_A)
=\frac{\Gamma(d+1)b^{a+d+1}\prod_{j=1}^d C_{A,j}}
{\Gamma(a+d+1)}\,t_A^a\exp(-bt_A),
\qquad a=0.2,\quad b=1,
\]
on $\zeta_{A,j}>0$ and $g_A>0$, where the constant is that of Lemma~1 of \citet{Gu2019} with $d$ inputs. It enters every leaf marginal likelihood, because $n_A$ and the spans in $C_{A,j}$ vary with the tree. In the coordinates $u_{A,j}=\log\phi_{A,j}$ and $v_A=\log g_A$, the Jacobian is $g_A\prod_j\zeta_{A,j}=\exp(v_A-\sum_ju_{A,j})$, so $\log\pi_{\mathrm{JR}}(\bfu_A,v_A)=\mathrm{const}+a\log t_A-bt_A-\sum_ju_{A,j}+v_A$.

In the sampler comparison of Section~4.2 and in Section~C.5, $g_A$ is bounded below by $10^{-6}$ for numerical stability. In Section~C.7, the same bound keeps $\Phi(A)$ finite for leaves whose training responses are all zero. The prior is then renormalized on $g_A\ge10^{-6}$, which replaces $\Gamma(a+d+1)/b^{a+d+1}=\int_0^\infty t^{a+d}\exp(-bt)\,dt$ in its normalizing constant by $\int_{10^{-6}}^\infty t^a(t-10^{-6})^d\exp(-bt)\,dt$. The nugget is optimized in $\log(g_A-10^{-6})$ instead of $\log g_A$, and $g_A-10^{-6}$ replaces $g_A$ in the Jacobian and in the nugget derivatives of Section~A.3. Section~C.5 also uses the bounds $10^{-8}$ and $10^{-4}$.

All simulation studies use the jointly robust prior. The storm-surge analysis uses independent inverse-gamma priors,
\[
\phi_{A,j}\sim\mathrm{IG}(2.5,1),
\qquad g_A\sim\mathrm{IG}(2,0.1),
\]
under which $\Phi(A)$ is finite for every response vector, including leaves whose training responses are all zero (Section~A.1). Both specifications integrate the leaf mean and process variance under the proper normal--inverse-gamma prior with $(m_0,V_0,a_0,b_0)=(0,100,2,1)$. For both storm-surge training sizes, each node retains at most 12 admissible split rules, allocated cyclically across admissible axes in input-column order with evenly spaced positions within each axis, and the axis and cut-point priors are normalized on the selected rules.
In these fits, child-node optimization starts from the fitted parent parameters, with bounds widened to contain them, and the gradient tolerance is $10^{-2}$. A leaf whose training responses are all equal is first started at long ranges and a small nugget, and then from the parent parameters. Pointwise predictions are used at both training sizes, because the covariance factorizations of joint predictions fail with 50,649 training runs. With 10,000 runs, joint and pointwise predictions give almost the same root mean squared error (RMSE) and continuous ranked probability score (CRPS).

\subsection{SPGP Initialization}\label{sc:spgp_init}
For the soft-partitioned Gaussian process (SPGP) of \citet{Luo2024}, we follow the coarse-to-fine initialization in the code of its authors. Local nearest-neighbor Gaussian process fits \citep{Datta2016} on a $6\times6$ pilot grid allocate reference knots, a fused-lasso step \citep{fusedlasso} on the knot graph gives an initial partition with $k_0$ regions, and each observation is softly assigned to its nearest knots with inverse-distance probabilities. We set $k_0=4$ in the two-dimensional example, matching its four regions, and $k_0=1$ in the smooth example of Section~C.3. Each region has an isotropic Mat\'ern range and variance ratio, re-estimated whenever the partition changes, and the partition is sampled by Markov chain Monte Carlo (MCMC) for 7,000 iterations, of which the first 2,000 are discarded.

\subsection{DGP VEC, TGP, and SVecchia Implementation}
For the two-layer Vecchia-approximated deep Gaussian process (DGP VEC), we use the \texttt{deepgp} package \citep{SauerDeepGP} with a Mat\'ern covariance and an inferred nugget. In the two-dimensional examples, \texttt{fit\_two\_layer} runs 6,000 iterations with conditioning sets of size 30, and the 6,000 continuation draws of \texttt{continue(..., 6000, re\_approx=TRUE)} are retained; the 10-dimensional fits use the same scheme with \texttt{D=10} and \texttt{pmx=TRUE}. The storm-surge fit with $n=10{,}000$ uses $m=70$ and \texttt{nmcmc=2000}, discards the first 1,000 iterations, retains every second remaining iteration, and fixes the outer nugget at $10^{-6}$. The response is standardized for fitting, and predictions are transformed back. The outer nugget contribution is removed from the predictive variance so that intervals target the latent response.

We fit TGP, the treed Gaussian process of \citet{Gramacy2012}, with the function \texttt{btgp} of the \texttt{tgp} package \citep[version 2.4-23;][]{Gramacy2007TGP}, using the package's Mat\'ern correlation with $\nu=2.5$, the smoothness used by TVGP. The two-dimensional fits use a constant mean and tree-prior parameters $(\alpha,\beta)=(0.95,0.5)$, and the 10-dimensional fits use the package defaults, a linear mean and $(\alpha,\beta)=(0.5,2)$. The predictive TGP fits use 7,000 MCMC iterations, of which the first 2,000 are discarded, and the two-dimensional fits retain every 20th iteration and the 10-dimensional fits every iteration.

For scaled Vecchia (SVecchia), we use the \texttt{fit\_scaled} and \texttt{predictions\_scaled} functions of \citet{KatzfussScaledVecchia}. The variance, nugget, and one range per input are estimated by Fisher scoring on the Vecchia likelihood, and the inputs are rescaled by the estimated ranges before ordering and selecting the conditioning sets. The conditioning sets for fitting, variance correction, and prediction all have size 200, with $\nu=2.5$ and \texttt{trend="pre"}, and the simulations also use \texttt{scale="parms"} and an estimated variance correction factor. Parameter estimation uses 5,000 observations in the two-dimensional example and in the storm-surge comparison with 10,000 runs, 2,000 in the 10-dimensional example, and 10,000 in the storm-surge comparison with 50,649 runs.

\subsection{JGP Implementation}
For the jump Gaussian process (JGP), we use the public JumpGaussianProcess implementation associated with \citet{Park2022} and \citet{Flowers2025}. Each prediction fits a local Gaussian process (GP) with a squared-exponential kernel and one lengthscale per input by classification expectation-maximization, with a quadratic boundary for two inputs and a linear boundary otherwise, on the 25 nearest training observations in the simulations. In the storm-surge comparison, 200 neighbors are selected from 25, 50, 100, 200, and 285 by CRPS, and then RMSE, on a 250-run validation subset of the 10,000-run training set of the first replication, and the response is standardized before fitting.
Intervals and CRPS target the latent function and use the generalized least-squares mean with the universal-kriging variance
\[
v_{\mathrm{JGP}}(\bfx_\ast)
=c_{\ast\ast}-k_\ast^\top K_y^{-1}k_\ast
+\frac{\left(1-\boldsymbol{1}^\top K_y^{-1}k_\ast\right)^2}
{\boldsymbol{1}^\top K_y^{-1}\boldsymbol{1}},
\]
where $K_y$ is the fitted local covariance matrix including the nugget, $k_\ast$ the covariance vector between the prediction point and its local training observations, and $c_{\ast\ast}$ the fitted latent variance at the prediction point.

\subsection{laGP Implementation}
For the local approximate Gaussian process (laGP) of \citet{laGP} in the storm-surge analysis, \texttt{laGP::laGP} is called separately at each test input, and \texttt{predGP(..., nonug=TRUE)} gives the latent predictive distribution. The response is standardized, and the inputs are divided by the square roots of separable lengthscales estimated with a nugget by maximum likelihood from a global GP fitted to 1,000 randomly selected training runs. Each local GP estimates its lengthscale and nugget by maximum likelihood on the 100 nearest training runs, a design selected over active-learning-Cohn designs \citep{Gramacy2015} grown from six to 50 or 100 points by CRPS, and then RMSE, on a 300-run validation subset of the 10,000-run training set of the first replication.

\subsection{Computing Resources}

Computing times are wall-clock times on a Linux computing cluster. During fitting, SMC evaluates the candidate splits of a node in parallel, while the particles are updated sequentially. In the storm-surge fits, the candidate splits are instead evaluated in one process, which uses 16 threads for the Vecchia computations and the conditioning-set searches. RJ-MCMC and TGP chains are sequential. An RJ-MCMC chain fits leaves in parallel only when a proposal creates two new leaves, and it reads most proposed leaves from a cache, so it runs mostly on one core. In Section~4.2, each SMC run, RJ-MCMC chain, and TGP chain is allocated 16 cores, and each fit in Section~4.3 is allocated four cores. In Section~4.4, TVGP is allocated four cores in the two-dimensional example and 16 in the 10-dimensional example, TGP 16 cores in both examples, DGP VEC 16 cores in the two-dimensional example and four in the 10-dimensional example, and SVecchia and JGP four cores in both examples. The TVGP fits of Sections~C.3 and~C.4 are allocated four cores. Each storm-surge TVGP fit is allocated 16 cores, and the SVecchia, laGP, and JGP fits run within one 16-core job per method.

\section{Additional Results} \label{sec: additional results}

\subsection{SMC Population Diagnostics}
\label{sec:sim1-smc-population-diagnostics}

We use the ten TVGP fits of the two-dimensional comparison in Section~4.4, one per dataset, each with $N=5{,}000$ particles started from the root-only tree. Two particles represent different trees if any split coordinate, cut value, child ordering, or terminal region differs, regardless of the covariance parameters. The fits terminate after 16--22 stages with zero to four resampling events (Figure~\ref{fig:supp-sim1-smc-diagnostics}(a)), and their final populations contain one to 18 distinct trees (panel (c)), with a median of eight. In each fit, the five highest-weight trees carry between 0.960 and 1.000 of the weight. The ten fits of the 10-dimensional comparison retain more trees, terminating after 21--25 stages with 28--189 distinct trees and final importance-weight effective sample size (ESS) fractions of 0.568--0.972.

\begin{figure}[!htbp]
    \centering
    \includegraphics[width=\linewidth,keepaspectratio]{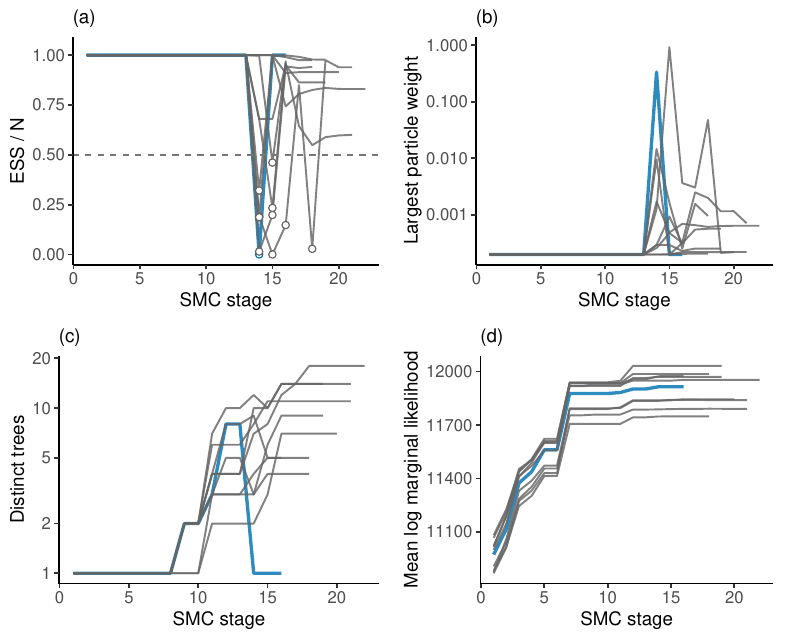}
    \caption{Stagewise diagnostics for the ten TVGP fits of the two-dimensional comparison, with $N=5{,}000$: (a) pre-resampling effective sample size (ESS) of the importance weights divided by $N$, (b) largest normalized particle weight before resampling, (c) number of distinct trees after any resampling on a logarithmic scale, and (d) particle-weighted mean log marginal likelihood before resampling. The fit to dataset 1 is blue, and the other nine are gray. Open circles in (a) mark resampling, and the dashed horizontal line marks $\mathrm{ESS}/N=1/2$. Resampling is rare, and the final populations contain few distinct trees.}
    \label{fig:supp-sim1-smc-diagnostics}
\end{figure}

For the two-dimensional populations, we also summarize the sizes of the terminal leaves, weighting each leaf by the weight of its tree and pooling the ten datasets. Leaf sizes range from 31 to 4,050 training observations, with a weighted median of 600, and about 65\% of the weighted leaves have more than 100 observations. Only about 1\% of the weighted leaves have at most $m+1=31$ observations, for which the Vecchia likelihood with conditioning sets of size $m=30$ is exact.

\subsection{Sensitivity to the Number of Particles}\label{sc:particle_sensitivity}

TVGP is run ten times with different seeds on the first dataset of the two-dimensional example of Section~4.4 at each of $N=30$, 50, 100, 250, 500, 1,000, 2,000, and 5,000, with the other TVGP settings of Section~B. Mean RMSE decreases from 0.01799 at $N=30$ to 0.01760 at $N=5{,}000$ and mean CRPS from 0.00882 to 0.00874, coverage of 95\% latent intervals stays between 0.9468 and 0.9479, and the mean scores change little above $N=100$. Mean computing time ranges from 8.8 to 11.3 minutes and does not increase monotonically with $N$, because candidate nodes are evaluated once and shared across particles.

\subsection{A Smooth Two-Dimensional Example}\label{sec:synth-2d}

The smooth function $f(x_1,x_2)=x_1\exp(-x_1^2-x_2^2)$ on $[-2,5]^2$ \citep{Gramacy2004}, observed with independent $\mathcal N(0,0.04^2)$ noise, has a localized peak and trough and no abrupt changes, unlike the two-dimensional example. Each of ten datasets has 6,000 training and 1,000 test inputs from two separately generated improved Latin hypercube designs, every method receives the same data within a dataset, and TVGP uses the settings of Section~B. TVGP has the second-largest mean RMSE of the six methods, and TGP has the smallest RMSE and CRPS but coverage 0.918 (Table~\ref{tab:additional-examples}).

\subsection{Oblique-Transition Version of the 10-Dimensional Example}\label{sec:ns10-rotated}

The 10-dimensional example has axis-aligned transitions in $x_1$. In this version, these two transitions occur in $(x_1+x_2)/2$, through $s_1=\operatorname{expit}[60\{(x_1+x_2)/2-0.30\}]$ and $s_2=\operatorname{expit}[60\{(x_1+x_2)/2-0.65\}]$, so they are oblique to the coordinate axes, while the remaining terms of the data-generating mechanism, the designs, the sample sizes, and the number of datasets follow Section~4.1, and the TVGP settings follow Section~B. TVGP keeps the lowest mean RMSE and CRPS under oblique transitions (Table~\ref{tab:additional-examples}), but its coverage is 0.906, whereas TGP, with intervals about twice as wide, has coverage closest to the nominal level. TVGP takes 2.8 hours on average, less than TGP and DGP VEC but more than SVecchia and JGP. The coverage of SVecchia and JGP is below 0.42.

\begin{table}[!htbp]
\centering
\caption{Mean scores in the smooth two-dimensional example of Section~C.3 and the oblique-transition 10-dimensional example of Section~C.4: RMSE, CRPS, coverage and mean width of 95\% latent intervals, and time for fitting and prediction. Means are over ten datasets, except for TGP and SPGP (eight) and DGP VEC (five) in the smooth example, for which the remaining fits are unavailable.}
\label{tab:additional-examples}
\vspace{2pt}
\begin{tabular}{lccccr}
\toprule
Method & RMSE & CRPS & Coverage & Width & Time\\
\midrule
\multicolumn{6}{l}{\emph{Smooth two-dimensional example}}\\
TVGP & 0.00837 & 0.00469 & 0.940 & 0.0318 & 25.9 min\\
TGP & 0.00527 & 0.00251 & 0.918 & 0.0169 & 11.7 h\\
SPGP & 0.00587 & 0.00405 & 0.992 & 0.0494 & 16.7 h\\
DGP VEC & 0.00675 & 0.00392 & 0.984 & 0.0368 & 7.5 h\\
SVecchia & 0.00648 & 0.00366 & 0.955 & 0.0265 & 7.7 min\\
JGP & 0.0109 & 0.00593 & 0.857 & 0.0312 & 0.8 min\\
\midrule
\multicolumn{6}{l}{\emph{Oblique-transition 10-dimensional example}}\\
TVGP & 0.468 & 0.239 & 0.906 & 1.39 & 2.8 h\\
TGP & 0.734 & 0.419 & 0.947 & 2.86 & 4.4 h\\
DGP VEC & 0.747 & 0.441 & 0.758 & 1.73 & 18.7 h\\
SVecchia & 0.655 & 0.456 & 0.419 & 0.740 & 33.8 min\\
JGP & 0.930 & 0.643 & 0.293 & 0.731 & 1.3 min\\
\bottomrule
\end{tabular}
\end{table}

\subsection{Sensitivity to the Jointly Robust Prior}\label{sec:jr-sensitivity}

The ten datasets of the two-dimensional example of Section~4.4 are refitted with the hyperparameters $a$ and $b$ of the JR prior and the lower bound on $g_A$ of Section~B.1 changed one at a time, from $a=0.2$, $b=1$, and $g_A\ge10^{-6}$ to $a\in\{0.05,0.5\}$, $b\in\{0.5,2\}$, or the lower bounds $10^{-8}$ and $10^{-4}$, with the other TVGP settings of Section~B. Predictive performance changes little across these seven settings. Mean RMSE ranges from 0.01934 to 0.01945, mean CRPS from 0.008821 to 0.008849, coverage from 0.9452 to 0.9466, and mean interval width from 0.06021 to 0.06039.

\subsection{Sampler Diagnostics}\label{sec:ns10-smc-rjmcmc}

The RJ-MCMC update used in the sampler comparison of Section~4.2 is given in Algorithm~\ref{alg:tvgp_rjmcmc_main}. Here $\mathcal G$ is the set of trees that satisfy the minimum leaf size and the depth limit, and $\widehat\Phi(A)$ is the leaf marginal likelihood evaluated by the deterministic Laplace approximation shared with SMC. Following \citet{Gramacy2012}, grow, prune, change, and swap moves are proposed with probabilities 0.20, 0.20, 0.40, and 0.20; a swap of a parent and child that split the same input uses a rotation, a swap of splits on different inputs exchanges the two rules, and proposals outside $\mathcal G$ produce self-transitions.

\begin{algorithm}[!htbp]
\caption{RJ-MCMC for the TVGP tree posterior}
\label{alg:tvgp_rjmcmc_main}
\begin{algorithmic}[1]
\Require Tree space $\mathcal G$, tree prior $\pi(T\mid X)$, leaf marginal likelihoods $\widehat\Phi(A)$, fixed move probabilities, and number of iterations $B$.
\State Define $\widetilde\ell(T)=\log\pi(T\mid X)+\sum_{A\in\mathcal L(T)}\log\widehat\Phi(A)$ and initialize $T^{(0)}\in\mathcal G$.
\For{$b=1,\ldots,B$}
    \State Draw a move type and propose $T^\star\sim q(\cdot\mid T^{(b-1)})$, using a rotation to swap a parent and child that split the same input.
    \If{$T^\star\notin\mathcal G$}
        \State Set $T^{(b)}=T^{(b-1)}$ and continue.
    \EndIf
    \State Evaluate $\widehat\Phi(A)$ for each leaf of $T^\star$ not evaluated before, and set
    \[
    \alpha=\min\left\{1,
    \exp\!\left[\widetilde\ell(T^\star)-\widetilde\ell(T^{(b-1)})\right]
    \frac{q(T^{(b-1)}\mid T^\star)}{q(T^\star\mid T^{(b-1)})}
    \right\}.
    \]
    \State Set $T^{(b)}=T^\star$ with probability $\alpha$; otherwise set $T^{(b)}=T^{(b-1)}$.
\EndFor
\end{algorithmic}
\end{algorithm}

Each sampler is run ten times on one dataset per example, with the TVGP settings of Section~B, maximum depth three, the screened starting values of Section~B.1, and candidate cut points that every node takes from one set of ten quantiles per input of all training inputs; SMC uses $N=5{,}000$ particles, and each RJ-MCMC chain starts from the tree with a single leaf and runs 50,000 iterations, of which the first 20,000 are discarded. After burn-in, nine of the ten two-dimensional RJ-MCMC chains make 328--640 partition transitions and the tenth makes two. Two of the chains remain at levels of the sum of leaf log marginal likelihoods about 14 and 46 below the level shared by the other eight chains and by all ten SMC runs, and the rank-normalized $\widehat R$ \citep{Vehtari2021} is 1.89 for this sum and 1.25 for the number of leaves. In the 10-dimensional example, four chains remain at one partition throughout the 30,000 post-burn-in iterations and the other six make two to 80 transitions (Figure~\ref{fig:current-rjmcmc-traces-10d}(b)). The final SMC populations contain two distinct trees in each two-dimensional run and two to four in each 10-dimensional run.

\begin{figure}[!htbp]
\centering
\includegraphics[width=\linewidth,keepaspectratio]{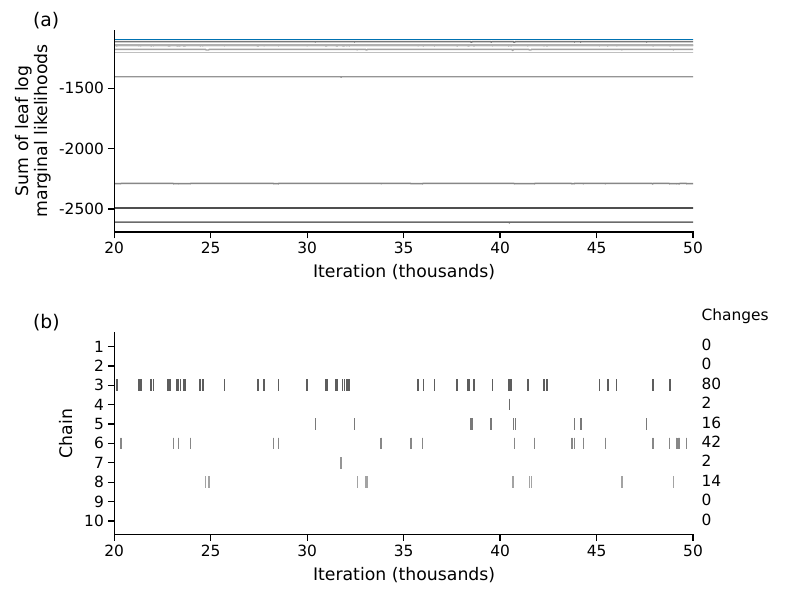}
\caption{RJ-MCMC chains after burn-in in the 10-dimensional example: (a) the sum of leaf log marginal likelihoods on a shared vertical scale and (b) training-partition transitions, with counts at right. Iterations 1--20,000 are discarded, chain 1 is blue. Four chains have no partition transitions.}
\label{fig:current-rjmcmc-traces-10d}
\end{figure}

In the 10-dimensional example, the mean final CRPS is 0.188 for SMC and 0.195 for the post-burn-in RJ-MCMC draws, and Section~4.2 gives the times and RMSE. Four of the cumulative RJ-MCMC estimates fall below the mean final SMC RMSE within 43 minutes, before the first SMC run finishes at 61.5 minutes. The ten SMC runs end with the same highest-weight tree, for which $\widetilde\ell(T)$ in Algorithm~\ref{alg:tvgp_rjmcmc_main} is $-1232.2$, whereas the final RJ-MCMC trees have $\widetilde\ell(T)$ between $-2641$ and $-1130$. The five chains that end with lower RMSE than SMC reach trees with larger $\widetilde\ell(T)$ than the SMC tree. In the sum of leaf log marginal likelihoods plotted in Figure~\ref{fig:current-rjmcmc-traces-10d}(a), which omits the log tree prior, these five chains end highest, a sixth ends slightly below them, and the other four end much lower.

Figure~\ref{fig:official-tgp-traces} shows ten TGP chains for each of the two datasets, each running 50,000 iterations of \texttt{tgp::btgp} with the first 20,000 discarded, a constant mean, $(\alpha,\beta)=(0.95,0.5)$, and at least 30 observations per leaf; these diagnostic chains are separate from the predictive TGP fits of Section~4.4 (Section~B.3). TGP updates the continuous GP parameters as well as the tree, and panels (a)--(b) show the package-reported log posterior and panels (c)--(d) the number of leaves. The rank-normalized $\widehat R$ of the log posterior is 2.43 in the two-dimensional example and 3.37 in the 10-dimensional example, with bulk ESS 12.1 and 11.0 and tail ESS 33.2 and 12.8, and $\widehat R$ for the leaf count is 6.08 and 3.21, so the chains remain at different tree sizes and posterior levels.

\begin{figure}[!htbp]
\centering
\includegraphics[width=\linewidth,keepaspectratio]{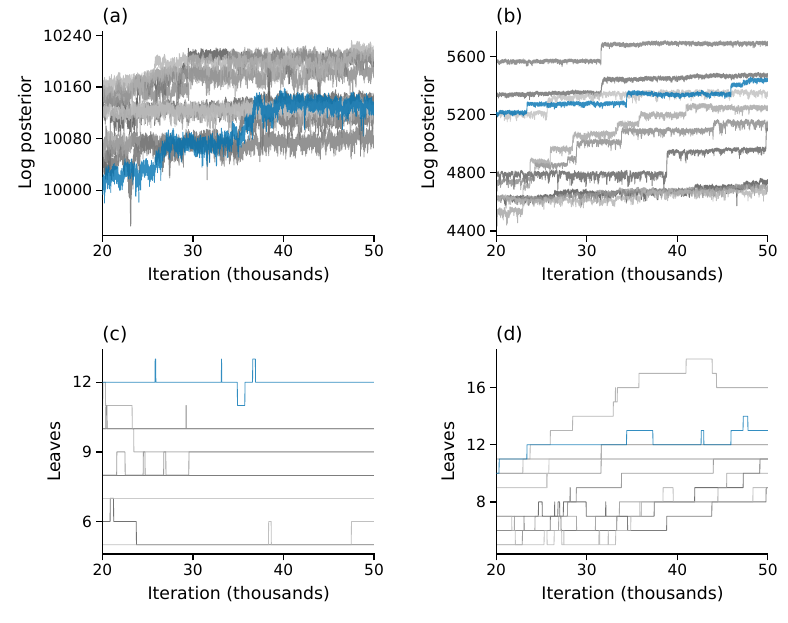}
\caption{TGP chains for the two-dimensional dataset (left) and the 10-dimensional dataset (right), with ten chains per dataset: (a)--(b) package-reported log posterior and (c)--(d) number of leaves. All panels show iterations 20,001--50,000. Within each panel, chains share a vertical scale; chain 1 is blue and the other nine are gray. The chains remain at different posterior levels and tree sizes.}
\label{fig:official-tgp-traces}
\label{fig:official-tgp-traces-10d}
\end{figure}

\subsection{Storm-Surge SMC and Calibration Diagnostics}
\label{sec:storm-diagnostics}

At each training size, TVGP is fitted once to each of the ten replications of Section~5.2, with the TVGP settings of Section~B. The fits with 10,000 training runs terminate after 23--40 stages and those with 50,649 runs after 133--177 stages, and Figure~6 shows their weight ESS ((a), (d)), distinct trees ((b), (e)), and final tree weights ((c), (f)). The largest final tree weight is 0.363--0.994 for the smaller fits and 0.216--0.998 for the larger fits, and in each fit where it is below 0.9, at least two trees have weight above 0.1.

Figure~\ref{fig:storm-current-regional-coverage} stratifies test runs by pressure-deficit and landfall-longitude quintiles of the pooled test inputs at each training size. Coverage varies more strongly with longitude (panels (b) and (d)). In the first three longitude quintiles, mean coverage ranges from 0.776 to 0.801 with 10,000 training runs and from 0.812 to 0.824 with 50,649, compared with about 0.95 and 1.00 in the last two quintiles. Across pressure-deficit quintiles (panels (a) and (c)), mean coverage ranges from 0.823 to 0.904 with 10,000 training runs and from 0.845 to 0.912 with 50,649. Pooled over replications, coverage is 0.889 for zero responses and 0.822 for positive responses with 10,000 training runs, and 0.908 and 0.838 with 50,649. The undercoverage therefore occurs in both response groups and persists with more training data.

\begin{figure}[!htbp]
\centering
\includegraphics[width=\linewidth,keepaspectratio]{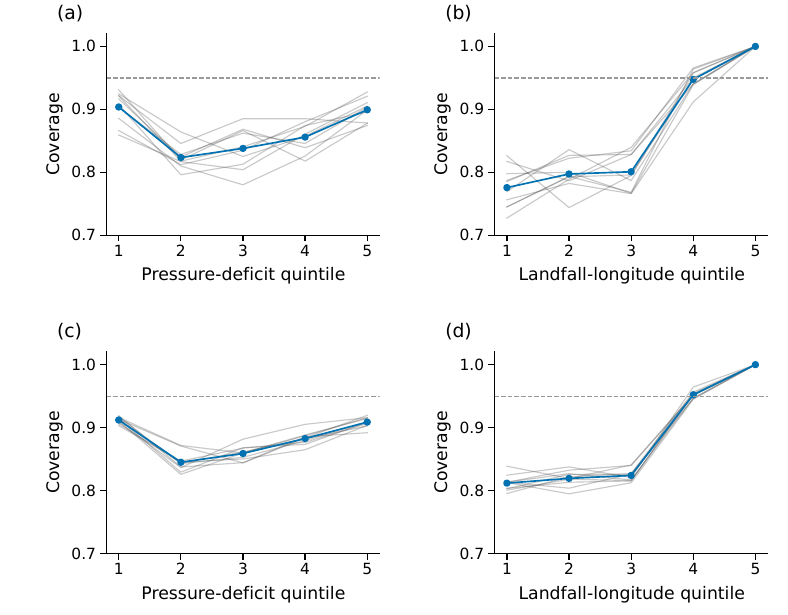}
\caption{TVGP coverage by pressure-deficit quintile, (a) and (c), and landfall-longitude quintile, (b) and (d): (a)--(b) ten replications with 10,000 training and 1,000 test runs and (c)--(d) ten replications with 50,649 training and 5,628 test runs. Gray lines show individual replications and blue lines their means. Quintiles are numbered in increasing order of the input, so longitude quintile 1 is the westernmost. Dashed lines mark 0.95 coverage. Coverage is lowest in the first three longitude quintiles.}
\label{fig:storm-current-regional-coverage}
\end{figure}

A separate grouped-design holdout assigns all landfall realizations of each five-dimensional storm design, given by the first five inputs in Table~2, to either the training or the test set, so that every test storm is new to the emulator. In each of ten replications, the held-out designs are chosen by a greedy maximin rule from a random starting design, and designs containing the two smallest or two largest values of an input stay in the training set. With about 10,000 training runs, 69 of 690 selected designs are held out, and with the full data, 350 of the 3,500 designs are held out, leaving about 50,650 training runs. With about 10,000 and about 50,650 training runs, mean coverage is 0.908 and 0.895 overall, 0.929 and 0.918 for exact-zero responses, and 0.866 and 0.852 for positive responses. With the full data, mean coverage in the five landfall-longitude quintiles is 0.837, 0.840, 0.841, 0.957, and 1.000, the same regional pattern as in the random holdout.

In all ten full-data fits of this holdout, some leaves contain both zero and positive training responses; we call them mixed leaves. For each test input, we draw 40 trees with replacement from the final weighted population and record the fraction of these trees that place the input in a mixed leaf. On average, 49.0\% of the test runs fall in a mixed leaf for more than half of the drawn trees, and their mean coverage is 0.802, compared with 0.985 for the other test runs. Among positive-response test runs, the corresponding coverages are 0.834 and 0.909, so the undercoverage is largest in leaves that straddle the boundary of the zero-response region.

\subsubsection{Exact-Zero Response Regions}

About 62\% of the storm-surge runs have zero peak surge elevation, and the mixed leaves above lie on the boundary of this zero-response region. To examine interval calibration near such a boundary in a controlled setting, we generate inputs in $[0,1]^2$ and set
\[
f(x_1,x_2)=\begin{cases}
0,&x_1\le q,\\
0.4+0.3\sin^2(2\pi x_2)+0.3(x_1-q)/(1-q),&x_1>q.
\end{cases}
\]
The two exact-zero scenarios use $q=0.3$ and $q=0.7$, with independent Gaussian noise of standard deviation 0.03 only in the active region. A third scenario uses $q=0.7$ and adds the same noise throughout the domain. Each scenario has ten datasets, each with 300 training inputs and 1,000 test inputs from two separately generated random Latin hypercube designs, so the datasets differ in their designs and noise. TVGP and TGP receive the same data, and predictions are evaluated against $f$. TVGP bounds $g_A$ below by $10^{-6}$ in all three scenarios (Section~B.1), which keeps $\Phi(A)$ finite for leaves whose training responses are all zero.

\begin{figure}[!htbp]
\centering
\includegraphics[width=\linewidth,keepaspectratio]{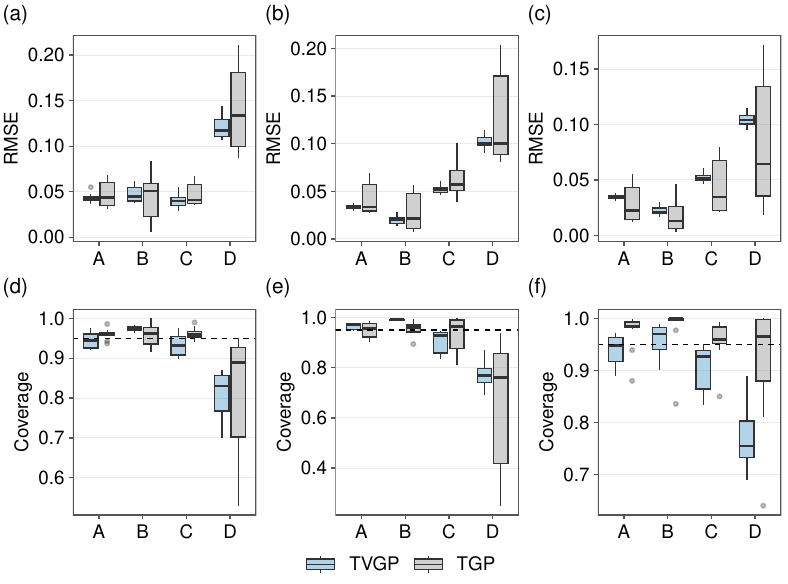}
\caption{Controlled inactive-region experiment. Columns correspond to 30\% exact zeros, 70\% exact zeros, and a 70\% inactive region with noise throughout the domain. Boxplots summarize ten paired datasets per scenario, each with 300 training and 1,000 test inputs; the top row, panels (a)--(c), gives RMSE and the bottom row, panels (d)--(f), coverage of 95\% latent intervals. Groups are A, all inputs; B, inactive region; C, active region; D, boundary strip $|x_1-q|\le0.05$. TVGP is blue and TGP gray; coverage axes differ across columns, and dashed lines mark 0.95. TVGP uses the jointly robust prior, $N=5{,}000$, and conditioning sets of size 30. TVGP coverage falls well below 0.95 in the boundary strip in all three scenarios.}
\label{fig:exact-zero-comparison}
\end{figure}

Mean overall TVGP coverage is 0.945 and 0.964 in the two exact-zero scenarios (Figure~\ref{fig:exact-zero-comparison}(d) and (e)), and the exact-zero regions themselves have high coverage for both methods. In the boundary strip $|x_1-q|\le0.05$, however, mean TVGP coverage is 0.810 and 0.768, compared with 0.807 and 0.644 for TGP. With noise throughout the domain (panel (f)), boundary coverage is 0.766 for TVGP and 0.917 for TGP. Good average coverage can therefore conceal substantial undercoverage near a discontinuity. For TVGP, this undercoverage arises from the discontinuity rather than from the exact zeros, because added noise leaves its boundary coverage nearly unchanged.

\subsection{Residual Semivariograms of the Storm-Surge Response}
\label{sc:additional_pairwise_heatmap}

Because the landfall locations lie along a curved coastline, we fit a principal curve through the unique locations \citep{Hastie01061989}, order them by the resulting coastal coordinate, and divide this ordering into five coastal bins, Coast~1 through Coast~5 (Figure~\ref{fig:supp-storm-coast-semivar}(a)). Residuals $e_i$ come from a generalized additive model fitted by restricted maximum likelihood with the R package \texttt{mgcv} \citep{Wood2011}, with tensor-product smooths in $(\Delta P,R_p)$ and in the landfall coordinates and univariate smooths in $V_f$, $\theta$, and $B$ (Table~2). Within each coastal bin, a random subsample of at most 1,200 runs is used, and the Euclidean distances between these runs, computed from the seven inputs each scaled to $[0,1]$, are divided into 15 intervals at their empirical quantiles. For each interval, the empirical residual semivariogram \citep[pp.~69--83]{Cressie1993} averages $(e_i-e_j)^2/2$ over the pairs whose distance falls in it, and it is plotted against the median distance in the interval.

\begin{figure}[!htbp]
\centering
\begin{minipage}[c]{0.54\linewidth}\centering\includegraphics[width=\linewidth,keepaspectratio]{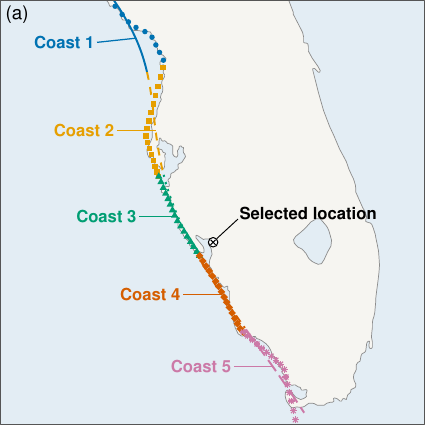}\end{minipage}\hfill\begin{minipage}[c]{0.42\linewidth}\centering\includegraphics[width=\linewidth,keepaspectratio]{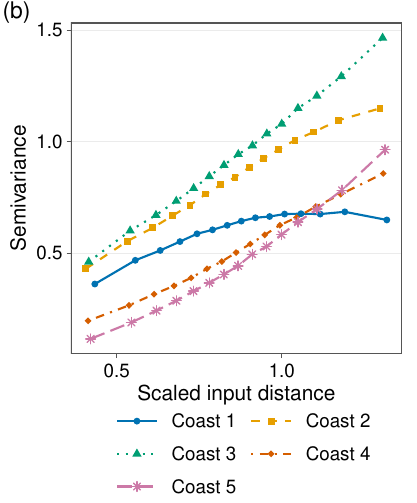}\end{minipage}
\caption{(a) Five coastal bins obtained by ordering unique landfall locations along a fitted principal curve, drawn as lines; symbols mark a subset of the landfall locations, and the labeled point marks the selected response location. (b) Empirical residual semivariograms, which differ in level and shape across the bins. Colors and symbols identify the coastal bins in both panels, and line types also distinguish them in (b).}
\label{fig:supp-storm-coast-semivar}
\end{figure}

The residual semivariograms in Figure~\ref{fig:supp-storm-coast-semivar}(b) differ in level and shape across the five coastal bins. The semivariogram for Coast 1 levels off at the largest distances, whereas those for Coasts 2 and 3 keep increasing and reach the largest values, and those for Coasts 4 and 5 start lowest and rise steadily over the whole distance range. Residual dependence therefore changes along the coastline after the additive trends are removed.

\needspace{6\baselineskip}
\section{Complete Predictive Comparisons and Ablation Scores}\label{sec:complete-predictive}

\subsection{Predictive Comparisons}

Table~\ref{tab:complete-scores} gives the mean scores of all methods in the two- and 10-dimensional comparisons of Section~4.4. The CRPS is smallest for TGP in the two-dimensional example and for TVGP in the 10-dimensional example. Among the methods with near-nominal coverage in the two-dimensional example, TGP has the narrowest intervals. In the 10-dimensional example, TGP has the coverage closest to the nominal level, with intervals more than twice as wide as those of TVGP.

\begin{table}[!htbp]
\centering
\caption{Mean scores in the two- and 10-dimensional comparisons of Section~4.4: RMSE, CRPS, coverage and mean width of 95\% latent intervals, and time for fitting and prediction. Means are over the ten datasets, except that the two-dimensional JGP means exclude one dataset whose CRPS and interval width diverged.}
\label{tab:complete-scores}
\vspace{2pt}
\begin{tabular}{lccccr}
\toprule
Method & RMSE & CRPS & Coverage & Width & Time\\
\midrule
\multicolumn{6}{l}{\emph{Two-dimensional example}}\\
TVGP & 0.0194 & 0.00883 & 0.946 & 0.0603 & 14.1 min\\
TGP & 0.0241 & 0.00733 & 0.957 & 0.0466 & 10.6 h\\
SPGP & 0.0286 & 0.0109 & 0.960 & 0.0862 & 12.0 h\\
DGP VEC & 0.0299 & 0.0127 & 0.945 & 0.0846 & 7.4 h\\
SVecchia & 0.0337 & 0.0138 & 0.908 & 0.0712 & 9.0 min\\
JGP & 0.0328 & 0.0133 & 0.858 & 0.0605 & 1.7 min\\
\midrule
\multicolumn{6}{l}{\emph{10-dimensional example}}\\
TVGP & 0.416 & 0.203 & 0.882 & 1.08 & 71.4 min\\
TGP & 0.789 & 0.424 & 0.947 & 2.76 & 2.7 h\\
DGP VEC & 0.797 & 0.460 & 0.780 & 1.89 & 19.0 h\\
SVecchia & 0.766 & 0.535 & 0.343 & 0.695 & 32.0 min\\
JGP & 0.980 & 0.665 & 0.302 & 0.758 & 1.2 min\\
\bottomrule
\end{tabular}
\end{table}

\begin{figure}[!htbp]
\centering
\includegraphics[width=0.85\linewidth,keepaspectratio]{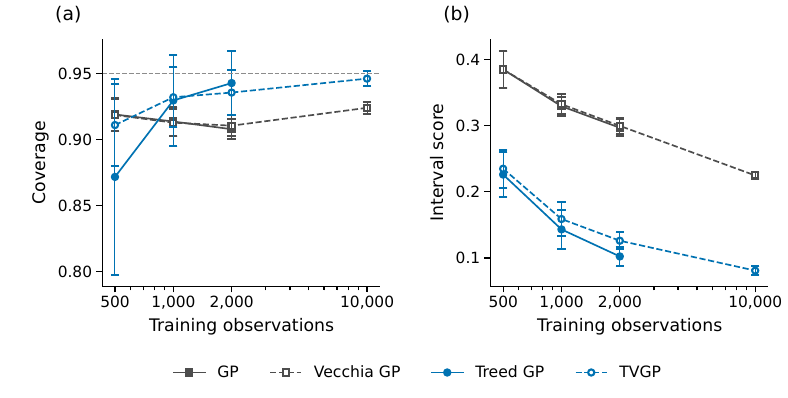}
\caption{Additional ablation scores: (a) coverage of 95\% latent intervals and (b) 95\% interval score. Points and bars show means and one standard deviation over the same ten paired datasets at each training size. Colors and line types are as in Figure~2, and the TVGP settings and specifications are those of Section~B and Section~4.3. The dashed horizontal line marks 0.95 coverage. The treed fits have smaller interval scores at every training size and, from $n=1{,}000$ on, coverage closer to 0.95.}
\label{fig:ablation-calibration}
\end{figure}

\subsection{Ablation Calibration}

The factorial design in Section~4.3 compares a single region with a learned tree, each under full-covariance and Vecchia likelihoods. All four specifications use the TVGP settings of Section~B. Figure~\ref{fig:ablation-calibration} adds the coverage of 95\% latent intervals and the 95\% interval score of \citet{Gneiting2007} to the comparison in Figure~2. For an interval $(l,u)$ and the latent value $f$ at a test input, the interval score is $(u-l)+40\{(l-f)\mathbb I(f<l)+(f-u)\mathbb I(f>u)\}$, so smaller values are better. At $n=2{,}000$, mean coverage in Figure~\ref{fig:ablation-calibration}(a) is 0.908 for GP, 0.911 for Vecchia GP, 0.943 for treed GP, and 0.936 for TVGP. At $n=10{,}000$, coverage is 0.924 for Vecchia GP and 0.946 for TVGP. At $n=2{,}000$ and $10{,}000$, the treed fits have less than half the mean interval score of the single-region fits (panel (b)).

\section{SMC Algorithm and Complexity Analysis}
\label{sec:smc_algorithm_complexity}
\begin{algorithm}[!htbp]
\small
\caption{Sequential Monte Carlo Inference for TVGP}
\label{alg:smc_tvgp_supp}
\begin{algorithmic}[1]
\Require Data $(\mathbf y,\mathbf X)$, particle count $N$, prior parameters, and candidate sets.
\Ensure Weighted set of trees $\{(T_i,W_i)\}_{i=1}^N$.

\State Set $t=0$; initialize each particle with the root-only tree $T_0^i$ and the list of active leaves $\mathcal U_0^i$ containing the root; set $W_0^i=1/N$.
\While{at least one particle has an active leaf}
    \State $t \leftarrow t + 1$
    \For{$i = 1$ to $N$}
        \State Copy $(T_{t-1}^i,\mathcal U_{t-1}^i)$ to $(T_t^i,\mathcal U_t^i)$ and set $w_t^i=1$.
        \If{$\mathcal U_t^i\ne\varnothing$}
            \State Remove the first leaf $A$ from $\mathcal U_t^i$.
            \If{$\mathcal J(A)=\varnothing$}
                \State Keep $A$ terminal.
            \Else
                \State Obtain $\widehat\Phi$ for $A$ and all admissible child pairs.
                \State Compute $\mathrm{BF}_A(j,c)$ from Equation~(7) and $Z_j(A)$ and $\kappa_A$ from Section~3.2.
                \State Draw $S(A)$ and, if $S(A)=1$, then $J(A)$ and $L(A)$ from $\mathbb{Q}_A$ in Proposition~3.4.
                \If{$S(A)=1$}
                    \State Split $A$ on axis $J(A)$ at $L(A)$ and append its left and right children to $\mathcal U_t^i$.
                \Else
                    \State Keep $A$ terminal.
                \EndIf
                \State Set $w_t^i=\kappa_A$.
            \EndIf
        \EndIf
        \State Set $\omega_t^i=W_{t-1}^iw_t^i$.
    \EndFor

    \State Normalize $W_t^i=\omega_t^i/\sum_k\omega_t^k$ and compute $\mathrm{ESS}=1/\sum_i(W_t^i)^2$.
    \If{$\mathrm{ESS}<N/2$}
        \State Draw $a_1,\ldots,a_N$ independently from $\operatorname{Categorical}(W_t^1,\ldots,W_t^N)$.
        \State Replace particle $i$ by a copy of particle $a_i$, including its tree and list of active leaves, for $i=1,\ldots,N$; set all weights to $1/N$.
    \EndIf
\EndWhile
\State \Return $\{(T_t^i, W_t^i)\}_{i=1}^N$.
\end{algorithmic}
\end{algorithm}

\paragraph{Computational complexity.}
For a node $A$ with $n_A$ training observations, each optimizer iteration for the Vecchia--Laplace approximation $\widehat\Phi(A)$ costs $\mathcal{O}(n_Am^3)$ for the Vecchia factors with conditioning sets of size $m$ and $\mathcal{O}(n_Ad^2m^2)$ for the Hessian of Proposition~3.3, a total of $\mathcal{O}\{n_Am^2(m+d^2)\}$. Evaluating the locally optimal proposal $\mathbb Q_A$ of Proposition~3.4 requires the marginal likelihoods of both children for at most $dK$ candidate splits $(j,c)$, where $d$ is the input dimension, and $K$ is the number of candidate cut points per input. Thus, the per-node cost is $\mathcal{O}\{dK\,n_Am^2(m+d^2)\}$ per optimizer iteration.
Nodes at a fixed depth are disjoint, so $\sum_{A:\,d_A=k} n_A\le n$. Let $D_T$ denote the largest realized depth among the $N$ particles. Then $\sum_A n_A\le n(D_T+1)$ for each particle, yielding the leading total cost
\[
\mathcal{O}\!\big\{N\,dK\,m^2(m+d^2)\,n(D_T+1)\big\}
\]
for $N$ particles, per optimizer iteration.
Weight normalization and resampling add $\mathcal{O}(N)$ work per SMC stage and are lower order in this setting. The bound ignores the reuse of cached node evaluations across particles, so observed computing times need not grow linearly in $N$ (Section~\ref{sc:particle_sensitivity}).
No depth cap is imposed, except in the sampler comparisons of Section~4.2 and Section~C.6, which use maximum depth three. The minimum leaf size bounds the size of the tree: if each child created by a split contains at least $n_{\min}$ observations, then the number of leaves is at most $\lfloor n/n_{\min}\rfloor$ and hence $D_T\le \lfloor n/n_{\min}\rfloor-1$. Algorithm~\ref{alg:smc_tvgp_supp} gives the complete SMC algorithm.

\FloatBarrier

\end{document}